\documentclass[12pt]{article}
\usepackage[bookmarks,bookmarksnumbered]{hyperref}
\usepackage{amsmath,XoohmE}
\usepackage{graphicx,color}
\usepackage{booktabs,multirow}
\usepackage{url,algorithm,algorithmic}

\newcommand{\codeand}{\mathbin{\&}} 
\newcommand{\codexor}{\boxplus} 
\newcommand{\codeplus}{\codexor} 

\newcommand{\wt}{\mathop{\text{\rm wt}}\nolimits}
\newcommand{\img}{\mbox{Im}}
\newcommand{\Aut}{\mbox{Aut}}
\newcommand{\cA}{\mathcal{A}}

\newcommand{\bF}{\relax\leavevmode\hbox{$F$\kern-.55em
                  \vrule height1.95ex depth-1.8ex width5.75pt}\kern1pt}

\newcommand{\ddt}{\partial_\tau}
\newcommand{\Cl}{\mathop{\text{Cl}}\nolimits}

\newcommand{\sM}{{\mathscr{M}}}

\newcommand{\fT}{\text{\sf T}}

\newcommand{\gen}{\mathbin{\,:\,}}
\newtheorem{theorem}{Theorem}[section]
\newtheorem{definition}{Definition}[section]
\newtheorem{proposition}[theorem]{Proposition}
\newtheorem{corollary}[theorem]{Corollary}

\long\def\oMit#1{}

\def\4{\hphantom{\hbox{$-$}}}
\def\8{\oplus}

\def\vC#1{\vcenter{\hbox{\hss#1\hss}}}
\definecolor{Hey}{rgb}{1,0,.4}
\definecolor{orange}{rgb}{.8,.4,0}
\definecolor{plum}{rgb}{.4,0,.6}
\definecolor{bone}{rgb}{.95,.95,.9}
\definecolor{gray}{gray}{.25}

\def\Ic{\hbox{\sf c\kern-3pt\rule[.2pt]{.5pt}{4pt}\kern2.5pt}}
\catcode`@=11    
\newcounter{xmpl}\resetby{section}{xmpl}
\newenvironment{example}{\par\noindent\addtocounter{xmpl}{1}%
                          \def\@currentlabel{{\thesection.\arabic{xmpl}}}%
                          {\bfseries Example~\arabic{section}.\arabic{xmpl}}%
                           ~\ignorespaces\small}
                        {{\color{gray}\hrulefill\rule{.45pt}{1ex}}\par}
\newcounter{cnst}\resetby{section}{cnst}
\newenvironment{construction}{\par\noindent\addtocounter{cnst}{1}%
                          \def\@currentlabel{{\thesection.\arabic{cnst}}}%
                          {\bfseries Construction~\thesection.\arabic{cnst}}%
                           ~\ignorespaces\slshape}
                        {\hrulefill\rule{.45pt}{1ex}\par}
\newcounter{crit}\resetby{section}{crit}

\def\Ft#1{\footnote{#1}}
\newdimen\parshift\parshift=\parindent
 \long\def\@footnotetext#1{\insert\footins{\reset@font\footnotesize\interlinepenalty%
  \interfootnotelinepenalty\splittopskip\footnotesep\splitmaxdepth\dp\strutbox%
   \floatingpenalty\@MM\hsize\columnwidth\addtolength{\hsize}{-2\parshift}
    \@parboxrestore\protected@edef\@currentlabel{\csname p@footnote\endcsname\@thefnmark}
      \color@begingroup
       \@makefntext{\rule\z@\footnotesep\ignorespaces#1\@finalstrut\strutbox\vglue1mm}
        \color@endgroup}}
 \long\def\@makefntext#1{\hglue\parshift
                         \vbox{\noindent\hb@xt@0em{\hss\@makefnmark\,}#1}\vglue2ex}
\catcode`@=12			
 \font\rOpe=cmsy10                        
 \def\ktl{{\hbox{\rOpe\char'170}}}        
 \def\kbl{{\hbox{\rOpe\char'170}}}        
 \def\kcr{{\reflectbox{\rOpe\char'170}}}        
 \def\ktr{{\reflectbox{\rOpe\char'170}}}        
 \def\kbr{{\reflectbox{\rOpe\char'170}}}        
 \def\Border{\vbox{\hsize0pt
        \setlength{\unitlength}{1mm}
        \newcount\xco
        \newcount\yco
        \xco=-21
        \yco=12
        \begin{picture}(0,0)(-7.5,0)
        \put(\xco,\yco){$\ktl$}
        \advance\yco by-1
        {\loop
        \put(\xco,\yco){$\kcr$}
        \advance\yco by-2
        \ifnum\yco>-240
        \repeat
        \put(\xco,\yco){$\kbl$}}
        \xco=170
        \yco=12
        \put(\xco,\yco){$\ktr$}
        \advance\yco by-1
        {\loop
        \put(\xco,\yco){$\kcr$}
        \advance\yco by-2
        \ifnum\yco>-240
        \repeat
        \put(\xco,\yco){$\kbr$}}
       \put(-19.5,13){\scalebox{.6125}{State University of New York
           Physics Department|University of Maryland Center for
           String and Particle  Theory \&\ Physics Department|%
           Howard University|Quid, Inc.}}
       \put(-19.5,-241.5){\scalebox{.768}{University of Alberta
           Mathematics Department|Pepperdine University Natural
           Sciences Division|Bard College Mathematics
           Program}}
        \end{picture}
        \par\vskip-8mm}}
\definecolor{UMred}{rgb}{.9,.05,.2}
 \def\UMbanner{\vbox{\hsize0pt
        \setlength{\unitlength}{.4mm}
        \thicklines
        \begin{picture}(0,0)(-30,-10)
        \put(165,16){\line(1,0){4}}
        \put(170,16){\line(1,0){4}}
        \put(180,16){\line(1,0){4}}
        \put(175,0){\line(1,0){4}}
        \put(180,0){\line(1,0){4}}
        \put(185,0){\line(1,0){4}}
        \put(169,0){\line(0,1){16}}
        \put(170,0){\line(0,1){16}}
        \put(179,0){\line(0,1){16}}
        \put(180,0){\line(0,1){16}}
        \put(184,0){\line(0,1){16}}
        \put(185,0){\line(0,1){16}}
        \put(169,16){\oval(8,32)[bl]}
        \put(170,16){\oval(8,32)[br]}
        \put(179,0){\oval(8,32)[tl]}
        \put(185,0){\oval(8,32)[tr]}
        \end{picture}
        \par\vskip-6.5mm
        \thicklines}}

 \SfTitles 
 \allowdisplaybreaks
 \seceq

\begin{document}
\thispagestyle{empty}
\vbox{\Border\UMbanner}
 \noindent
 \today \hfill{UMDEPP 08-010, SUNY-O/667}
  \vfill
 \begin{center}
{\LARGE\sf\bfseries\boldmath
  Codes and Supersymmetry in One Dimension}\\[3mm]
  \vfill
{\sf\bfseries C.F.\,Doran$^a$, M.G.\,Faux$^b$, S.J.\,Gates, Jr.$^c$,
     T.\,H\"{u}bsch$^d$,\\ K.M.\,Iga$^e$, G.D.\,Landweber$^f$, and R.L.\,Miller$^{g,h}$}\\[3mm]
{\small\it
  $^a$Department of Mathematical and Statistical Sciences,\\[-1mm]
      University of Alberta, Edmonton, Alberta, T6G 2G1 Canada%
  \\[-4pt] {\tt  doran@math.ualberta.ca}
  \\
  $^b$Department of Physics,
      State University of New York, Oneonta, NY 13825%
  , {\tt  fauxmg@oneonta.edu}
  \\
  $^c$Center for String and Particle Theory,\\[-1mm]
      Department of Physics, University of Maryland, College Park, MD 20472%
  , {\tt  gatess@wam.umd.edu}
  \\
  $^d$Department of Physics \&\ Astronomy,\\[-1mm]
      Howard University, Washington, DC 20059
  \\[-4pt] {\tt  thubsch@howard.edu}
  \\
  $^e$Natural Science Division,
      Pepperdine University, Malibu, CA 90263%
  , {\tt  Kevin.Iga@pepperdine.edu}
  \\
 $^f$Mathematics Program, Bard College,
     Annandale-on-Hudson, NY 12504-5000%
  , {\tt  gregland@bard.edu}
  \\
 $^g$Department of Mathematics,\\[-1mm]
      University of Washington, Seattle, WA 98105%
  , {\tt  rlmill@math.washington.edu}
 \\
 $^h$Quid, Inc.,
 733 Front Street, C1A\\[-1mm]
 San Francisco, CA 94111%
 . {\tt rmiller@quid.com}
 }\\[3mm]
  \vfill
{\sf\bfseries ABSTRACT}\\[3mm]
\parbox{145mm}{Adinkras are diagrams that describe many useful supermultiplets in $D=1$ dimensions.  We show that the topology of the Adinkra is uniquely determined by a doubly even code.  Conversely, every doubly even code produces a possible topology of an Adinkra.  A computation of doubly even codes results in an enumeration of these Adinkra topologies up to $N=28$, and for minimal supermultiplets, up to $N=32$.}
  
\end{center}
  \vfill\vfill

\clearpage
\setcounter{page}{0}
\pagenumbering{arabic}
\section{Introduction, Review and Synopsis}
 \label{IRS}
Although many supersymmetric theories have been known since the 1970s, there is still no overarching classification of supermultiplets, even in one dimension (time).  In fact, supersymmetry in one dimension has been the subject of several investigations, for instance the development of the ${\cal GR}(d,N)$ algebras\cite{rGR1,rGR2,rGLP}, the development of Adinkras\cite{rA}, and to other efforts \cite{rFKS,rFS1,rFS2,rKRT,rBKMO,rBG,rBKLS,KT1,GRT,KT2}.

An Adinkra is a directed graph with various colorings and other markings on vertices and edges, which in a pictoral way encode all details of the supersymmety transformations on the component fields within a supermultiplet in one dimension.\cite{rA}  The main purpose of the present work is to determine the kinds of graphs (i.e., the {\em\/topology\/}) that can be used for Adinkras.

The class of supermultiplets described by Adinkras is wide enough to contain many noteworthy superfields.\cite{rA,MatterN2Matter}  This paper will clarify the conditions for a supermultiplet in one dimension to be described using an Adinkra.  Beyond this, by showing that this class is large, we thereby show that the class of supermultiplets is also large, and thus establish that the classification problem of supersymmetry in one dimension is much more intricate than it might appear at first.  We also thereby add many new supermultiplets to the literature that were previously unknown, and it is possible that some of these newly discovered supermultiplets may be useful or interesting in their own rights.

Finally, we hope that classifying Adinkras may help one better understand the conditions under which a superfield has an off-shell description.  Subject to a particular set of assumptions about dynamics, Siegel and Ro\v{c}ek had previously shown\cite{SR} that not all supermultiplets have off-shell descriptions.  On the other hand, superfields described by Adinkras are off-shell supermultiplets.  Therefore, understanding the range of what Adinkras can describe may shed some light on the question of which supermultiplets admit an off-shell description.

\subsection{Adinkras}
In one dimension, with $N$ supersymmetry generators $Q_1,\ldots,Q_N$, the supersymmetry algebra is
\begin{equation}
\{Q_I,Q_J\}=2i\delta_{IJ} \ddt
\end{equation}
where $\ddt$ is the derivative in the time direction.

An Adinkra is a finite directed graph, with every vertex colored either white or black, and with every edge colored one of $N$ colors (each color corresponds to one of the supersymmetries $Q_I$), and each edge drawn with either a solid or a dashed line.  The vertices correspond to the component fields (black for fermions, white for bosons) and the edges correspond to the action of each of the $Q_I$, in a way that is reminiscent of the Cayley diagram of a finitely generated group, or even more analogously, the Schreier diagram of the set of cosets of a subgroup.   Details of Adinkras and how they correspond to supermultiplets can be found in Refs.~\cite{rA,r6-1}.  The classification of Adinkras naturally falls into four steps:
\begin{enumerate}
\item Determine which topologies are possible (the topology of an Adinkra is the underlying graph of vertices and edges without colorings, as, for instance, in Ref.\cite{r6-1}).
\item Determine the ways in which vertices and edges may be colored.  The topology of the Adinkra, together with the colorings of vertices and edges, will be called the {\em chromotopology} of the Adinkra.  It is chromotopologies that are classified in this paper.
\item Determine the ways in which edges may be chosen as dashed or solid.  This is closely related to the well-known theory of Clifford algebras, and will be studied in a future effort.
\item Determine the ways in which arrows may be directed along each edge.  This issue is addressed in Ref.\cite{r6-1}, and shown to be equivalent to the question of ``hanging'' the graph on a few sinks.  Alternately, we can start with an Adinkra where all arrows go from bosons to fermions, then perform a sequence of vertex raises to arrive at other choices of arrow directions.
\end{enumerate}

As it happens, it is convenient to do 1.~and 2.~together; that is, to classify chromotopologies.  Herein, we show that the classification of Adinkra chromotopologies is equivalent to another interesting question from coding theory: the classification of doubly even codes.  Much work has already been done in this area\cite{rMcWS,rCHVP,rCPS}, and the work described in this paper goes even further in developing this classification; see Appendix~\ref{app:RM}.

We emphasize that we focus here on the representation theory, not the dynamics. This is natural, as we need to first know the full palette of supersymmetric representations before discussing the properties of the dynamics in theories built upon such representations. For instance, presupposing a standard, uncoupled Lagrangian for the supermultiplets that we intend to classify would necessarily limit the possibilities; there do exist supermultiplets which can only have interactive Lagrangians\cite{rHSS,rGSS}. Herein, we defer the task of finding Lagrangians involving the supermultiplets considered in this paper. In Refs.\cite{r6-2,r6-7a}, we have in fact started on such studies, and, using Adinkras, have constructed supersymmetric Lagrangians for some of the supermultiplets that are also discussed herein.

In units where $\hbar=1=c$, all physical quantities may have at most units of mass, the exponent of which is called the {\em\/engineering dimension\/} and is an essential element of physics analysis in general. The engineering dimension of a field $\f(\t)$ will be written $[\f]$; for more details, see Refs.\cite{r6-1,r6--1}.

\subsection{Main Result}
 \label{sMR}
Our main result about the chromotopology types of Adinkras and the corresponding supermultiplets, up to direct sums, may be summarized as follows:

We define the function:
\begin{equation}
 \vk(N):=
  \begin{cases}
   0 &\text{for $N<4$},\\
   1 &\text{for $N=4,5$}, \\
   2 &\text{for $N=6$}, \\
   3 &\text{for $N=7$}, \\
   4 + \vk(N{-}8) &\text{for $N\geq 8$, recursively}.
  \end{cases}
 \label{eKmax}
\end{equation}
\begin{enumerate}\itemsep=-3pt\vspace{-3mm}
 \item Every Adinkra can be separated into its connected components.  (The supermultiplet corresponding to such an Adinkra breaks up into a direct sum of other supermultiplets, each of which corresponds to one of the connected components of the Adinkra).
 \item There is a one-to-one correspondence between possible chromotopologies of connected Adinkras and doubly even codes of length $N$.
 \begin{enumerate}\itemsep=-3pt
  \item Each connected chromotopology has, associated to it, a doubly even code of length $N$ and dimension $k\le \vk(N)$ that records which paths connect a vertex to itself.
  \item The chromotopology is then the quotient of the colored $N$-dimensional cube by this code. (The colored $N$-cube is the set of vertices and edges of the $N$-dimensional cube $[0,1]^N$, with colors on the edges determined by which axis it is parallel to, and colors of vertices according to the number, modulo 2, of coordinates that are 1).
  \item This quotient can be viewed as an iterated $k$-fold $\ZZ_2$-quotient.
  \item These chromotopologies really do come from supermultiplets in $D=1$ dimension, and if the arrows are chosen properly (one-hooked) we can arrange it so that it is easy to see that different codes give rise to different supermultiplets.
  \item Permuting the columns of a code corresponds to permuting the colors of the chromotopology, which in turn describes $R$-symmetries of the supermultiplet.
  \item There are an enormous multitude of distinct doubly-even codes for $N\le 32$, even when counting permutation equivalent codes as the same code.  Thus, there is an enormous multitude of Adinkra chromotopologies.
 \end{enumerate}
\end{enumerate}\vspace{-3mm}

This paper is organized as follows:
 Section~\ref{s:Codes} is a brief introduction to codes, and Section~\ref{adinkras} provides a review of Adinkras and their relationship with supermultiplets.  The first major result, in Section~\ref{AdiTop}, is that each Adinkra chromotopology gives rise to a doubly even code.  It will be convenient to provide a few classes of examples of doubly even codes for our discussions, and to give a sense for how many doubly even codes there are, so this is done in Section~\ref{s:codes-redux}.  We then turn to the second major result: that every doubly even code actually arises as the code for an Adinkra chromotopology for a supermultiplet.  This is done in Sections~\ref{s:ClifAd} and \ref{s:construct}.  Sections~\ref{s:fourd} and \ref{CliffHanger} discuss some consequences and directions for further research.

\section{Codes}
 \label{s:Codes}
We begin with a brief introduction to the theory of codes.  For a more thorough introduction to the subject, see Refs.\cite{rMcWS,rCHVP,rCPS}.

We think of $\{0,1\}$ as a group with the operation $\codexor$, which is addition modulo 2, i.e., the group $\ZZ_2$.  For the purposes of this paper, a {\em\/code} of length $N$ means a subgroup of $\{0,1\}^N$.\footnote{In the coding literature, there are sometimes other more general definitions of codes.  What we have described is a linear binary block code of length $N$.}
Though the standard notation for an element of a cartesian product is $(x_1,x_2,\cdots,x_N)$, in practice we frequently abandon the parentheses and the commas, so that the element $(0,1,1,0,1)$ may be written more succinctly as the codeword $01101$.  The components of such an $N$-tuple are called bits, and the $N$-tuple is called a word.  This word is called a {\em\/codeword} if it is in the code.

Now, $\ZZ_2$ is not only a group; it is also a field, so $\{0,1\}^N$ can be viewed as a vector space over $\{0,1\}$.  All the concepts of linear algebra then apply, but with $\mathbb{R}$ replaced with $\ZZ_2$.  Elements of $\{0,1\}^N$ may be thought of as vectors, with vector addition the operation $\codexor$ of bitwise addition modulo 2.  Codes are then linear subspaces of $\{0,1\}^N$.    Every code has a basis, called a {\em\/generating set}, $g_1,\ldots,g_k$, so that every codeword can be written uniquely as a sum
\begin{equation}
\sum_{i=1}^k x_i g_i,
\end{equation}
where the coefficients $x_1,\ldots,x_k$ are each either $0$ or $1$.  The number $k$ is the same for every generating set for a given code, and is called the {\em\/dimension} of the code.  It is common to say we have an $[N,k]$ linear code when $N$ is the length of the codewords and $k$ is the dimension.  It is traditional to denote a generating set as an $k{\times}N$ matrix, where each row is an element of the generating set.

If $v\in\{0,1\}^N$, we define the {\em\/weight} of $v$, written $\wt(v)$, to be the number of $1$s in $v$.  For instance, the weight of $01101$ is $\wt(01101)= 3$.

A code is called {\em\/even} if every codeword in the code has even weight.  It is called {\em\/doubly even} if every codeword in the code has weight divisible by $4$.  Examples of doubly even codes are given in Section~\ref{s:deC} below.

If $v$ and $w$ are in $\{0,1\}^N$, then $v\codeand w$ is defined to be the ``bitwise and'' of $v$ and $w$: the $i$th bit of $v\codeand w$ is $1$ if and only if the $i$th bit of $v$ and the $i$th bit of $w$ are both $1$.  A basic fact in $\{0,1\}^N$ is
\begin{equation}
\wt(v \codexor w)=\wt(v)+\wt(w)-2\,\wt(v\codeand w).\label{eqn:inclusionexclusion}
\end{equation}
There is a standard inner product.  If we write $v$ and $w$ in $\{0,1\}^N$ as $(v_1,\ldots,v_N)$ and $(w_1,\ldots,w_N)$, then
\begin{equation}
\langle v,w\rangle \equiv \sum_{i=1}^N v_i\,w_i\pmod{2}.\label{eqn:codeinnerproduct}
\end{equation}
We call $v$ and $w$ orthogonal if $\langle v,w\rangle=0$.  This occurs whenever there are an even number of bit positions where both $v$ and $w$ are $1$.  Note that $\langle v,v\rangle\equiv\wt(v)\pmod{2}$, and thus, when $\wt(v)$ is even, $v$ is orthogonal to itself.  Also note that $\langle v,w\rangle \equiv \wt(v\codeand w)\pmod{2}$.  One important consequence for us is that if $\wt(v)$ and $\wt(w)$ are multiples of 4, then  (\ref{eqn:inclusionexclusion}) implies that $\wt(v \codexor w)$ is a multiple of 4 if and only if $v$ and $w$ are orthogonal.

\section{Supersymmetric Representations and Adinkras}
 \label{adinkras}
The $N$-extended supersymmetry algebra without central charges in one dimension is generated by the time-derivative, $\ddt$, and the $N$ supersymmetry generators, $Q_1, \ldots, Q_N$, satisfying the following supersymmetry relations:
\begin{equation}
 \big\{\,Q_I\,,\,Q_J\,\big\}=2\,i\,\d_{IJ}\,\ddt,\quad
  \big[\,\ddt\,,\,Q_I\,\big] =0,\quad  I,J=1,\ldots,N. \label{eSuSy}
\end{equation}
In this section we determine some essential facts about the transformation rules of these operators on fields for which it is possible to maintain the physically motivated concept of engineering dimension.
We note that since the time-derivative has engineering dimension $[\ddt]=1$, the supersymmetry relations\eq{eSuSy} imply that the engineering dimension of the supersymmetry generators is $[Q_I]=\inv2$.

\subsection{Supermultiplets as Representations of Supersymmetry}
A real supermultiplet $\sM$ is a real, finite-dimensional, linear representation of the algebra\eq{eSuSy}, in the following sense: It is spanned by a basis of real bosonic and fermionic {\em\/component fields\/}, $\f_1(\t),\ldots,\f_m(\t)$ and $\j_1(\t),\ldots,\j_m(\t)$, respectively; each component field is a function of time, $\t$. The supersymmetry transformations, generated by the Hermitian operators $Q_1,\cdots,Q_N$, act linearly on $\sM$ while satisfying Eqs.\eq{eSuSy}
.
 The supermultiplet is {\em\/off-shell\/} if no differential equation is imposed on it\Ft{Logically, it is possible for some---but not all---component fields to become subject to a differential equation. This does not violate the literal definition of the off-shell supermultiplet. However, it does obstruct standard methods of quantization, which is our eventual purpose for keeping supermultiplets off-shell. For an example in 4-dimensional supersymmetry, see Ref.\cite{MatterN2Matter}.}.
 The number of bosons as fermions is then the same, guaranteed by supersymmetry.

\subsection{Building Supermultiplets from Adinkras}
 \label{sCloSuSy}
Refs.\cite{rA,r6-1,r6-2} introduced and then studied Adinkras, diagrams that encode the transformation rules of the component fields under the action of the supersymmetry generators $Q_1,\ldots,Q_N$.

 Supermultiplets that can be described by Adinkras have a collection of bosonic and fermionic component fields and a collection of supersymmetry generators $Q_1,\dots,Q_N$, so that: ({\bf1})~Given a bosonic field $\f$ and a supersymmetry generator $Q_I$, the transformation rule for $Q_I$ of $\f$ is of the form
\begin{alignat}{3}
 \text{either}&&\qquad Q_I \f &= \pm\, \j, \label{eQB1}\\
 \text{or}&&\qquad Q_I \f &= \pm\, \ddt \j, \label{eQB2}
\intertext{for some fermionic field $\j$. ({\bf2})~Given instead a fermionic field $\eta$ and a supersymmetry generator $Q_I$, the transformation rule of $Q_I$ on $\eta$ is of the form}
 \text{either}&&\qquad Q_I \eta &= \pm\, i\,B, \label{eQF1}\\
 \text{or}&&\qquad Q_I \eta &= \pm\, i\,\ddt B, \label{eQF2}
\end{alignat}
for some bosonic field $B$.  In particular, these supersymmetry generators act linearly using first-order differential operators.  Furthermore, the supersymmetry algebra requires that
\begin{alignat}{3}
 Q_I \f &= \pm\,  \j \qquad&\Longleftrightarrow\qquad
 Q_I \j &= \pm\, i\,\ddt \f, \label{eQBF1}
\intertext{and}
 Q_I \f &= \pm\, \ddt\j \qquad&\Longleftrightarrow\qquad
 Q_I \j &= \pm\, i\,\f,\label{eQBF2}
\end{alignat}
and where the $\pm$ signs are correlated to preserve Eqs.\eq{eSuSy}.

More generally, suppose we label the bosons $\f_1,\dots,\f_m$ and the fermions $\j_1,\dots,\j_m$.  Choose an integer $I$ with $1\le I\le N$, and an integer $A$ with $1\le A\le m$.  For each such pair of integers, we consider the transformation rules for $Q_I$ on the boson $\f_A$, and we might expect that these will be of the form
\begin{equation}
Q_I \, \f_A(\t)= c\,\ddt^{\l}\, \j_B(\t),\label{eQB}
\end{equation}
where $c=\pm 1$, $\l=0$ or $1$, and $B$ is an integer with $1\le B\le m$, so that $\j_B$ is some fermion; each of $c,\l,B$ will, in general, depend on $I$ and $A$.  Note that
\begin{equation}
\l=[\f_A]-[\j_B]+\inv2,
 \label{eED}
\end{equation}
for $\f_A$ and $\j_B$ to have a definite engineering dimension---provided the transformation rules had only dimensionless constants as we assume throughout. For each such transformation rule, we will get a corresponding transformation rule for the $Q_I$ on the fermion $\j_B(\t)$ that looks like this:
\begin{equation}
Q_I \, \j_B(\t) = \frac{i}{c}\,\ddt^{1-\l}\, \f_A(\t).\label{eQF}
\end{equation}
\Eqs{eQB}{eQF} constitute all of the transformation rules on the bosons and fermions, respectively.

\begin{definition}\label{dAd}
A supermultiplet, $\sM$, is {\em\/adinkraic\/} if all of its supersymmetric transformation rules are of the form\eq{eQB} and\eq{eQF}.

 For each adinkraic supermultiplet, its {\em Adinkra\/}, $\cA_\sM$, is a directed graph, consisting of a set of vertices, $V$, a set of edges, $E$, a coloring $C$ of the edges, a set of their orientations, $O$, and a labeling $D$ of each edge corresponding to whether or not it is dashed. 

Each component field of $\sM$ is represented by a vertex in $\cA_\sM$: white for bosonic fields and black for fermionic ones, thus equipartitioning the vertex set $V\to W$.  Every transformation rule of the form\eq{eQB} is depicted by an edge connecting the vertex corresponding to $\f_A$ to the vertex corresponding to $\j_B$, and color the edge with the $I^\text{th}$ color.  We use a dashed edge if $c=-1$, and oriented it from $\f_A$ to $\j_B$ if $\l=0$ and the other way around if $\l=1$.
\end{definition}

Table~\ref{t:A} illustrates the four possibilities for an edge.
\begin{center}
\begin{table}[ht]
  \centering
  \begin{tabular}{@{} cc|cc @{}}
    \makebox[15mm]{\bf Adinkra} & \makebox[40mm]{\bf\boldmath$Q$-action} 
  & \makebox[15mm]{\bf Adinkra} & \makebox[40mm]{\bf\boldmath$Q$-action} \\ 
    \hline
    \begin{picture}(5,9)(0,5)
     \put(0,0){\includegraphics[height=11mm]{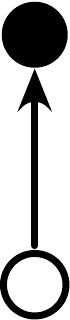}}
     \put(3,0){\scriptsize$A$}
     \put(3,9){\scriptsize$B$}
     \put(-1,4){\scriptsize$I$}
    \end{picture}\vrule depth4mm width0mm
     & $Q_I\begin{bmatrix}\j_B\\\f_A\end{bmatrix}
           =\begin{bmatrix}i\dot\f_A\\\j_B\end{bmatrix}$
  & \begin{picture}(5,9)(0,5)
     \put(0,0){\includegraphics[height=11mm]{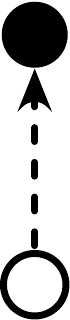}}
     \put(3,0){\scriptsize$A$}
     \put(3,9){\scriptsize$B$}
     \put(-1,4){\scriptsize$I$}
    \end{picture}\vrule depth4mm width0mm
     & $Q_I\begin{bmatrix}\j_B\\\f_A\end{bmatrix}
           =\begin{bmatrix}-i\dot\f_A\\-\j_B\end{bmatrix}$ \\[5mm]
    \hline
    \begin{picture}(5,9)(0,5)
     \put(0,0){\includegraphics[height=11mm]{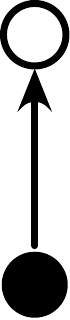}}
     \put(3,0){\scriptsize$B$}
     \put(3,9){\scriptsize$A$}
     \put(-1,4){\scriptsize$I$}
    \end{picture}\vrule depth4mm width0mm
     &  $Q_I\begin{bmatrix}\f_A\\\j_B\end{bmatrix}
           =\begin{bmatrix}\dot\j_B\\i\f_A\end{bmatrix}$
  & \begin{picture}(5,9)(0,5)
     \put(0,0){\includegraphics[height=11mm]{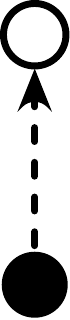}}
     \put(3,0){\scriptsize$B$}
     \put(3,9){\scriptsize$A$}
     \put(-1,4){\scriptsize$I$}
    \end{picture}\vrule depth4mm width0mm
     &  $Q_I\begin{bmatrix}\f_A\\\j_B\end{bmatrix}
           =\begin{bmatrix}-\dot\j_B\\-i\f_A\end{bmatrix}$ \\[5mm]
    \hline
  \multicolumn{4}{l}{\vrule height3.5ex width0pt\parbox{120mm}{\small The edges are here labeled by the variable index $I$; for fixed $I$, they are drawn in the $I^{\text{th}}$ color.}}
  \end{tabular}
  \caption{The correspondences between the Adinkra components and supersymmetry transformation formulae:
    vertices\,$\iff$\,component fields;
    vertex color\,$\iff$\,fermion/boson;
    edge color/index\,$\iff$\,$Q_I$;
    edge dashed\,$\iff$\,$c=-1$; and
    orientation\,$\iff$\,placement of $\ddt$.
    They apply to all $\f_A,\j_B$ within a supermultiplet and all $Q_I$-transformations amongst them.}
  \label{t:A}
\end{table}
\end{center}

We can also use the Adinkra to reconstruct the adinkraic supermultiplet, since the Adinkra contains all the information necessary to write down the transformation rules. Thus, an Adinkra is simply a graphical depiction of the transformation rules\eq{eQB1}--\eq{eQF2}.

\section{Adinkra Chromotopologies}
 \label{AdiTop}
It is the purpose of this paper to classify the possible {\em\/topologies\/} for Adinkras. To this end, we will need the precise definition:

\begin{definition}\label{dA2T}
The {\em\bfseries\/topology\/} of an Adinkra, $\fT(\cA_\sM)$, is the graph $(V,E)$ consisting of only the (unlabeled) vertices and edges of the Adinkra; cf.\ Definition~\ref{dAd}. Also, $\fT(\sM):=\fT(\cA_\sM)$.
\end{definition}
In particular, from Definition~\ref{dAd} of an Adinkra, we forget the bipartition (black or white) of the vertices and all additional information associated to the edges, namely, the color of the edges (which records the $Q_I$ corresponding to that edge), dashedness ($c=\pm1$), and direction (the exponent of $\ddt$ in \eqs{eQB1}{eQF2}).
\begin{definition}
The {\em\bfseries\/chromotopology\/} of an Adinkra is the topology of the Adinkra, together with the vertex bipartition (coloring each vertex black or white), and the edge coloring (assigning a color to each edge).
\end{definition}
In particular, from Definition~\ref{dAd} of an Adinkra, we forget the dashedness of the edges and the direction of the arrow along the edge.

Suppose we have an adinkraic multiplet with component fields $(F_1\6(\t),\cdots,F_{2m}\6(\t))$,\footnote{Previously, we labeled the bosons $\f_1,\cdots,\f_m$ and the fermions $\j_1,\cdots,\j_m$. Here, it will be convenient for notation to treat them on the same footing.\label{introf}} and let $\cA$ be its Adinkra.  We wish to consider the chromotopology.  That is, we wish to ignore the arrows on the edges and ignore whether an edge is dashed.  We can capture this information by taking the vertex set of the Adinkra $V=\{v_1,\cdots,v_{2m}\}$, the coloring of these vertices, and the edge set $E$ with its coloring.  Since this graph is inherited from an Adinkra, for every vertex and every $I\in\{1,\dots,N\}$ there is an edge corresponding to applying $Q_I$ to the field corresponding to that vertex.  Applying this $Q_I$ may involve derivatives, a sign, and/or a factor of $i$, but it will be convenient for now to suppress this.  To this end, we define functions $q_1,\cdots,q_N$ from the vertex set $V=\{v_1,\cdots,v_{2m}\}$ to itself, such that whenever $A$, $B$, and $I$ are such that there is an equation of type (\ref{eQB}) or (\ref{eQF}),
\begin{equation}
  Q_I \, F_A(\t)= c\,\ddt^{\lambda} F_B,\qquad\implies\qquad
  q_I(v_A)=v_B.
\label{e:qAct}
\end{equation}
Notice that in defining $q_I$ from $Q_I$, we are forgetting the coefficients $c=\pm1$, as well as the $\ddt$'s which encode the differences in engineering dimensions.
 The supersymmetry algebra (\ref{eSuSy}) then implies:
\begin{alignat}{3}
  q_I^2&=\Ione,\qquad&\text{\ie}\qquad
  q_I\big(q_I(v)\big)&=v,~\text{for all $v\in V$},\label{e:q2=1}\\[-8mm]
\intertext{and\vspace{-4mm}}
  q_Iq_J&=q_Jq_I,\qquad&\text{\ie}\qquad
  q_I\big(q_J(v)\big)&=q_J\big(q_I(v)\big),~\text{for all $I,J$, and all $v\in V$}.
   \label{e:qiqj}
\end{alignat}

The vertex set $V=\{v_1,\ldots,v_{2m}\}$, the partition of $V$ into bosons and fermions, and the maps $q_I:V\to V$ then encode the chromotopology.

\subsection{The Colored $N$-cube Chromotopology}\label{s:cubical}
The fundamental example of an Adinkra topology is that of the $N$-cube, $I^N=[0,1]^N$. It has $2^N$ vertices and $N{\cdot}2^{N-1}$ edges. We may embed it in $\IR^N$ by locating the vertices at the points $\vec{p}=(p_1,\cdots,p_N)\in\IR^N$, where $p_I=0,1$ in all $2^N$ possible combinations.  An edge connects two vertices that differ in precisely one coordinate. For every vertex, $\vec{p}$, the {\em\/weight of $\vec{p}$}, written $\wt(\vec{p})$, equals the number of $J\in\{1,\cdots,N\}$ for which $p_J=1$.

There is a natural chromotopology associated to the $N$-cube: we color vertices of even weight white, and vertices of odd weight we color black.  We associate the numbers from $1$ to $N$ with $N$ different colors, and then color each edge of the $N$-cube that connects vertices which differ only in the $I$th coordinate with the $I$th color.  The result is called the colored $N$-cube.
 
Using the notation from the previous section, we determine $q_1,\ldots, q_N$ as follows:
\begin{equation}
q_I(p_1,\ldots,p_{I-1},p_I,p_{I+1},\ldots,p_N)
=(p_1,\ldots,p_{I-1},1-p_I,p_{I+1},\ldots,p_N).
\end{equation}

There are indeed Adinkras of supermultiplets that have this chromotopology, as we will see in Section~\ref{s:construct}.

\subsection{The Quotient of a Colored Cube by a Code}
\label{s:quotient}
We now introduce another chromotopology: given a colored $N$-cube and a code $C$ of length $N$, we will create a chromotopology called $[0,1]^N/C$.

Before we do so, recall that $C$ is a subset of $\{0,1\}^N$, and so codewords (elements of $C$) are vectors which we can write as $\vec{g}$.  Likewise, the vertices of the colored $N$-cube are elements of $\{0,1\}^N$ as well.

For each vertex $\vec{v}\in\{0,1\}^N$, we identify it with all vertices of the form $\vec{v}\codexor \vec{g}$ for $\vec{g}\in C$.  The resulting vertex set we call $V'$.  Likewise, consider an edge $e$ colored $I$, connecting $\vec{v}$ to $\vec{w}$ (note that this means the $\vec{v}$ and $\vec{w}$ are identical except in bit $I$, where they differ).  We identify $e$ with the other edges colored $I$ that connect $\vec{v}\codexor \vec{g}$ with $\vec{w}\codexor \vec{g}$ (note that $\vec{v}\codexor \vec{g}$ and $\vec{w}\codexor \vec{g}$ also differ precisely at bit $I$, and so there is a unique edge connecting them, and it is colored $I$).  One way to view this is that the group $C$ acts on the vertex set and the edge set by $\codexor$, and we are taking the orbit space of this group action.

Now note that we have thereby created a graph with vertices and edges, and the colors of the edges is determined by this construction.  In order to consistently label the colors of the vertices, we must only identify vertices when they are of the same color (boson or fermion).  Since the color of a vertex in the colored $N$-cube is determined by its weight modulo 2, and since \Eq{eqn:inclusionexclusion} says that 
\begin{equation}
\wt(v \codexor g)=\wt(v)+\wt(g)-2\,\wt(v\codeand g),
\end{equation}
we see that $\wt(v\codexor g)$ agrees with $\wt(v)$ modulo 2 for all $g$ if and only if $\wt(g)$ is always even.  Therefore, if $C$ is an even code, this quotient of a colored $N$-cube gives a consistent coloring of the vertices.

This quotient may not necessarily be the chromotopology of the Adinkra of a supermultiplet: we will see in the next section that for this, the code must be doubly even.  We will then see in Section~\ref{s:construct} that this is sufficient.

\subsection{The Code of an Adinkra Chromotopology}
 \label{sQuot}
Suppose an Adinkra is not connected.  Then the corresponding supermultiplet splits into a direct sum, each component of which corresponds to a connected component of the Adinkra.  So if we can determine the chromotopologies of connected Adinkras, the general chromotopology of an Adinkra will be a finite disjoint union of such.  Thus we focus our attention to connected Adinkras.

\begin{construction}\label{const:thecode}
Given a connected Adinkra with $N$ edge colors, we construct a code of length $N$ as follows.

Let $V$ be the set of vertices.  Define $q_I:V\to V$ as above.  Pick any vertex $v_*\in V$ and fix it.  Now define
\begin{equation}
C=\left\{(x_1,\ldots,x_N)\,|\,q_1^{x_1}\big(\cdots q_N^{x_N}(v_*)\big)=v_*\right\}.
\end{equation}
\end{construction}

\begin{proposition}
The set $C$ defined in this construction is a code.
\end{proposition}

\begin{proof}
We need to prove that the elements of $C$ form a subgroup of $\{0,1\}^N$.  To prove closure, let $\vec{x}=(x_1,\ldots,x_N)$ and $\vec{y}=(y_1,\ldots,y_N)$ be elements of $C$.  By \Eq{e:qiqj}, we have
\begin{equation}
q_1^{x_1\codexor y_1}\cdots q_N^{x_N\codexor y_N}(v_*)
=q_1^{x_1}\cdots q_N^{x_N} \big(q_1^{y_1}\cdots q_N^{y_N}(v_*)\big)=v_*
\end{equation}
and so $\vec{x}\codexor\vec{y}\in C$.

To prove identity, note that $q_0^{0}\cdots q_N^0 (v_*)=v_*$.

To prove inverse, note that in $\{0,1\}^N$, every element is its own inverse.
\end{proof}

\begin{proposition}
The code $C$ is independent of the choice of vertex $v_*$.\label{prop:codeinvariant}
\end{proposition}
\begin{proof}
Let $v_0$ be any other vertex in the Adinkra.  Since the Adinkra is connected, there is a sequence of edges connecting $v_*$ to $v_0$.  If the colors of these edges are $I_1, \ldots, I_m$, then by the definition of $q_I$, we have
\begin{equation}
q_{I_1}\cdots q_{I_m}(v_*)=v_0
\end{equation}
Using \Eq{e:qiqj}, we can rearrange the $q_{I_i}$ so that they are in increasing order, and we can furthermore cancel pairs of $q_{I_i}$ so that each $q_I$ appears at most once.  In this way we obtain a vector $\vec{t}=(t_1,\ldots,t_N)\in\{0,1\}^N$ so that
\begin{equation}
q_1^{t_1}\cdots q_N^{t_N}(v_*)=v_0.\label{e:veezero}
\end{equation}
Now suppose $\vec{g}$ is a codeword obtained by applying Construction~\ref{const:thecode} using $v_*$.  Then 
\begin{equation}
q_1^{g_1}\cdots q_N^{g_N}(v_*)=v_*.
\end{equation}
Apply $q_1^{t_1}\cdots q_N^{t_N}$ to both sides, using \Eq{e:qiqj} to commute the $q_I$ past each other, and we get
\begin{equation}
q_1^{g_1}\cdots q_N^{g_N}(q_1^{t_1}\cdots q_N^{t_N}(v_*))=q_1^{t_1}\cdots q_N^{t_N}(v_*).
\end{equation}
Then we plug in \Eq{e:veezero} and get
\begin{equation}
q_1^{g_1}\cdots q_N^{g_N}(v_0)=v_0.
\end{equation}
Therefore any codeword for $v_*$ is a codeword for $v_0$.  Applying the argument with $v_0$ and $v_*$ switched proves the converse.
\end{proof}

\Remk Thus we can determine the code of an Adinkra as follows: first pick any vertex $v_*$.  Then for each $\vec{g}=(g_1,\ldots,g_N)\in\{0,1\}^N$, write down the set $S$ of all colors $I$ for which $g_I=1$, in some order.  Then take the path that starts at $v_*$ and goes along edges according to the colors in the set $S$.  If we return to $v_*$, then $\vec{g}$ is in the code.

For example, consider the Adinkra in \eq{A:D4} in Section~\ref{s:construct} below.  If we take $v_*=\phi_{(0000)}$, we can follow the colors black (color 1), then red (color 2), then green (color 3), then blue (color 4).  This takes us on a path that goes to $\psi_{(1000)}$ then $\phi_{(1100)}$, then $\psi_{(1110)}$, then back to $\phi_{(0000)}$.  From this we see that $1111$ is in the code.  On the other hand, suppose we took red (color 2), then blue (color 4).  This takes us from $\phi_{(0000)}$ to $\psi_{(0100)}$ to $\phi_{(1010)}$, which is not $\phi_{(0000)}$.  So we see that $0101$ is not in the code.  You can check that the code is $\{0000,1111\}$ no matter which starting vertex $v_*$ we use.

\begin{theorem}
Every connected Adinkra chromotopology is isomorphic to a quotient of a colored $N$-dimensional cube by the code of the chromotopology.
\label{T:QuoC}
\end{theorem}
\begin{proof}
Suppose we have a connected Adinkra with $N$ edge colors, with vertex set $V$ and edge set $E$.  Pick any one bosonic vertex $v_*\in V$ and fix it.  As in Construction~\ref{const:thecode}, we define
\begin{equation}
C=\{(g_1,\ldots,g_N)\,|\,q_1^{g_1}\cdots q_N^{g_N}(v_*)=v_*\}.
\end{equation}
We then take the quotient of the colored $N$-cube by $C$.  Let the vertex set of this quotient be called $W$, and let the edge set be called $F$.  Now we produce an isomorphism from the chromotopology of the quotient of the colored $N$-cube to the chromotopology of the original Adinkra.  By this we mean a bijection from $W$ to $V$ that preserves the colors of the vertices, and a bijection from $F$ to $E$ that preserves the colors of the edges and the relationship of which edges are incident with which vertices.

To define the bijection $f:W\to V$, let $w\in W$ be any vertex in $W$.  It comes from a quotient of an $N$-cube, meaning it is the result of identifying vectors in $\{0,1\}^N$, so there is at least one vector $(w_1,\ldots,w_N)\in\{0,1\}^N$ that was used in the identification to produce $w$.  We define
\begin{equation}
f(w)=q_1^{w_1}\cdots q_N^{w_N}(v_*).
\end{equation}
We note that any other choice of $N$-tuple that produces $w$ would be $(w_1,\ldots,w_N)\codexor \vec{g}$ where $\vec{g}\in C$, and
\begin{equation}
q_1^{w_1\codexor g_1}\cdots q_N^{w_N\codexor g_N}(v_*)
=q_1^{w_1}\cdots q_N^{w_N}\big(
q_1^{g_1}\cdots q_N^{g_N}(v_*)\big)
=q_1^{w_1}\cdots q_N^{w_N}(v_*)=f(w)
\end{equation}
so that $f$ does not depend on the choice of $(w_1,\ldots,w_N)$.

We now prove $f$ is injective.  Suppose $f(w)=f(x)$.  Then
\begin{equation}
q_1^{w_1}\cdots q_N^{w_N}(v_*)=q_1^{x_1}\cdots q_N^{x_N}(v_*).
\end{equation}
By using \Eq{e:qiqj}, we get
\begin{equation}
q_1^{w_1\codexor x_1}\cdots q_N^{w_N\codexor x_N}(v_*)=v_*,
\end{equation}
which implies $(w_1,\ldots,w_N)\codexor(x_1,\ldots,x_N)=(g_1,\ldots,g_N)$ is in the code $C$.  Then $(x_1,\ldots,x_N)=(w_1,\ldots,w_N)\codexor(g_1,\ldots,g_N)$, so that the vertices $(w_1,\ldots,w_N)$ and $(x_1,\ldots,x_N)$ are identified in the construction of the quotient.  Therefore $w=x$.

To see that $f$ is surjective, let $v\in V$ be any vertex.  Recall that we are assuming that the Adinkra is connected.  That is, every vertex $v\in V$ is connected, via a path of edges, to the fixed $v_*\in V$.  In the Adinkra, these edges have colors, forming a sequence, $I_1,\cdots,I_k$, when tracing from the vertex $v$ to $v_*$. If we then apply to $v_*$ a corresponding sequence of $q_I$'s, we get:
\begin{equation}
  q_{I_1}\cdots q_{I_k}(v_*)=v.
\end{equation}
Using the commutativity of the $q$'s to put them in numerical order and \Eq{e:q2=1} to eliminate those that appear more than once, we can write this as
\begin{equation}
q_1^{x_1}\cdots q_N^{x_N}(v_*) = v.
\end{equation}
But the left hand side is just $f(x_1,\ldots,x_N)$.  Therefore $f$ is surjective.

We now examine the edges.  Let $e$ be an edge in $F$ with color $I$.  It connects two vertices $v$ and $w$, so that $q_I(v)=w$.  We wish to show that there is an edge in $E$ colored $I$ connecting $f(v)$ to $f(w)$; that is, that $q_I(f(v))=f(w)$.

Choose a vector $(v_1,\ldots,v_N)\in\{0,1\}^N$ that gets used in the identification to make $v$.  For $(w_1,\ldots,w_N)$ we take
\begin{equation}
(w_1,\ldots,w_N)=(v_1,\ldots,v_{I-1},1-v_I,v_{I+1},\ldots,v_N)
\end{equation}
which is possible because this is the edge of color $I$ coming from $(v_1,\ldots,v_N)$ in the original $N$-cube before quotienting.  Now we apply $q_I$ to $f(v)$:
\begin{align}
q_I(f(v))&=q_I\big(q_1^{v_1}\cdots q_N^{v_N}(v_*)\big)\\
&=q_1^{v_1}\cdots q_{I-1}^{v_{I-1}}q_I^{1-v_I}q_{I+1}^{v_{I+1}}\cdots q_N^{v_N}(v_*)\\
&=q_1^{w_1}\cdots q_N^{w_N}(v_*)\\
&=f(w)
\end{align}
Therefore there is an edge colored $I$ connecting $f(v)$ to $f(w)$.

The fact that $\p$ sends bosons to bosons and fermions to fermions can now be seen by the fact that it sends the boson $(0,\dots,0)\in\{0,1\}^N$ to the boson $v_*$, and the fact that every vertex in $V$ is connected to $v_*$ by a sequence of edges, each of which alternates between bosons and fermions.
\end{proof}

We next consider what kinds of codes $C$ can appear in Adinkras.  We have already noted in the previous section that in order to get a consistent vertex-coloring, $C$ must be an even code.  We now show that the edge-dashing of the original Adinkra (which we have been ignoring in studying chromotopologies) will force $C$ to be {\em\/doubly\/} even; that is, for each $\vec{g}\in C$, $\wt(\vec{g})$ must be a multiple of 4.

\begin{theorem}
The code for an Adinkra must be doubly even.
\end{theorem}

\begin{proof}
To prove this, the $q_I$'s no longer suffice, and we will need to use the $Q_I$'s; in particular, we need to ``remember'' the scaling constants $c$ encoding the edge-dashedness in an Adinkra, and the equipartition of the vertices into bosonic and fermionic ones. 

The statement that $\vec{x}\in C$ means that
\begin{equation}
q_1^{x_1}\big(\cdots q_N^{x_N}(v_*)\big)=v_*,
\end{equation}
which means, if $F_*$ is the component field that is represented by $v_*$, then
\begin{equation}
 Q_1^{x_1}\cdots Q_N^{x_N}\, F_*(\t) = c\,\ddt^{\,\wt(\vec{x})/2}\,F_*(\t)
 \label{eBack*}
\end{equation}
for some complex number $c$. The exponent of $\ddt$ follows simply from comparing engineering dimensions of the left-hand side and the right-hand side.

Since this sequence of $Q_I$ operators, corresponding to a closed path in the Adinkra, must send bosons to bosons and fermions to fermions, it must be that $\wt(\vec{x})$ is even, so that $\inv2\wt(\vec{x})$ is indeed an integer and \Eq{eBack*} is well-defined.

Applying $Q_1^{x_1}\cdots Q_N^{x_N}$ twice to $F_*(\t)$, we find:
\begin{align}
 Q_1^{x_1}\cdots Q_N^{x_N}\cdot Q_1^{x_1}\cdots Q_N^{x_N} F_*(\t)
  &= c^2 \ddt^{\wt(\vec{x})}F_*(\t). \label{eCC0}
\intertext{On the left side, using the supersymmetry algebra\eq{eSuSy}, we can anti-commute the $Q_I$ past each other. Rearrange these to regroup the result into}
(-1)^{\wt(\vec{x})\choose2}Q_1^{2x_1}\cdots Q_N^{2x_N}F_*(\t)
 &= c^2 \ddt^{\wt(\vec{x})}F_*(\t), \label{eCC1}
\intertext{which, using the supersymmetry algebra\eq{eSuSy} again, becomes}
 (-1)^{\wt(\vec{x})\choose2}\,i^{\wt(\vec{x})} \,\ddt^{\wt(\vec{x})} F_*(\t)
 &= c^2 \ddt^{\wt(\vec{x})}F_*(\t).
  \label{eCC2}
\end{align}
Comparing the exponents of $\ddt$ in \Eq{eCC2} confirms the exponent of $\ddt$ in \Eq{eBack*}.
Comparing the numerical coefficients produces:
\begin{align}
 c^2=(-1)^{\wt(\vec{x})\choose2}\,i^{\wt(\vec{x})}
    =i^{\wt(\vec{x})(\wt(\vec{x})-1)+\wt(\vec{x})}
    =i^{\left(\wt(\vec{x})^2\right)}.
\end{align}
Since $\wt(\vec{x})$ is even, we know that $\wt(\vec{x})^2$ is a multiple of 4.  Thus, the left hand side of this is 1.

Using (\ref{eQB}) and (\ref{eQF}) repeatedly, we see that $c$ is $\pm 1$ if $\wt(\vec{x})=0\pmod4$, and $\pm i$ if $\wt(\vec{x})=2\pmod4$. But $c^2=1$ implies that $c=\pm 1$, and this can happen only if $\wt(\vec{x})=0\pmod4$.
\end{proof}

\Remk
The converse is also true.   That is, if $G$ is a doubly even code, then there is a family of adinkraic supermultiplets, the chromotopology of the Adinkra of which is a quotient by $G$ of a colored $N$-cube.  This will be done in Section~\ref{s:construct}.

\section{Finding Doubly Even Codes}\label{s:codes-redux}
\subsection{Examples of Doubly Even Codes}
 \label{s:deC}
Since doubly even codes classify chromotopologies, it is useful to consider a few examples of such codes.  For each $N$ there is a trivial doubly even code $\{00\cdots 0\}$ with one element, which we call $t_N$; its generating set is the empty set. In addition, when $N=4$, there is a code $\{0000,1111\}$, called $d_4$.  The generating set is $\{1111\}$.  More generally, for every even $N\ge 4$, there is a doubly even code called $d_N$, of length $N$ and with $\frac{N}{2}-1$ generators, with generating set
\begin{equation}
\begin{bmatrix}
1\,1\,1\,1\,0\,0\,0\,0\,0\,\cdots\,0\,0\,0\,0\,0\\[-1mm]
0\,0\,1\,1\,1\,1\,0\,0\,0\,\cdots\,0\,0\,0\,0\,0\\[-1mm]
0\,0\,0\,0\,1\,1\,1\,1\,0\,\cdots\,0\,0\,0\,0\,0\\[-1mm]
 \qquad\quad\vdots\\[-1mm]
0\,0\,0\,0\,0\,0\,0\,0\,0\,\cdots\,0\,1\,1\,1\,1
\end{bmatrix}.
\end{equation}
Note that this is a description of the generating set, so that the actual code has more codewords, including the null-vector and all those constructed by adding (bitwise, modulo 2) any number of these generators together. For example,
\begin{equation}
 \begin{bmatrix}
 1\,1\,1\,1\,0\,0\,0\,0\\[-1mm]
 0\,0\,1\,1\,1\,1\,0\,0\\[-1mm]
 0\,0\,0\,0\,1\,1\,1\,1\\[-1mm]
 \end{bmatrix}
\qquad\text{generates}\qquad
 \left\{\begin{array}{rc}
        v_0=&0\,0\,0\,0\,0\,0\,0\,0\\[-1mm]
        v_1=&1\,1\,1\,1\,0\,0\,0\,0\\[-1mm]
        v_2=&0\,0\,1\,1\,1\,1\,0\,0\\[-1mm]
        v_3=&0\,0\,0\,0\,1\,1\,1\,1\\[-1mm]
        v_1\codeplus v_2=&1\,1\,0\,0\,1\,1\,0\,0\\[-1mm]
        v_2\codeplus v_3=&0\,0\,1\,1\,0\,0\,1\,1\\[-1mm]
        v_1\codeplus v_3=&1\,1\,1\,1\,1\,1\,1\,1\\[-1mm]
        v_1\codeplus v_2\codeplus v_3=&1\,1\,0\,0\,0\,0\,1\,1\\[-1mm]
        \end{array}\right\}.
 \label{ed8}
\end{equation}
Note that the same code, on the right-hand side of the display\eq{ed8}, is just as well generated by $\{v_1,v_2,(v_1\codeplus v_3)\}$ and several other choices. For general $N$, the $d_N$ code contains $2^{\frac{N}2-1}$ codewords.

When $N$ is congruent to $7$ or $8$ modulo $8$, there is an important doubly even code called $e_N$, the generating set of which is that of $d_N$ (or $t_1\oplus d_{N-1}$ when $N\equiv 7 \pmod{8})$ augmented by an additional generator of the form $101010\cdots$.  For instance,
\begin{equation}
e_7\,:~
\begin{bmatrix}
1\,1\,1\,1\,0\,0\,0\\[-1mm]
0\,0\,1\,1\,1\,1\,0\\[-1mm]
1\,0\,1\,0\,1\,0\,1
\end{bmatrix},\qquad
e_8\,:~
\begin{bmatrix}
1\,1\,1\,1\,0\,0\,0\,0\\[-1mm]
0\,0\,1\,1\,1\,1\,0\,0\\[-1mm]
0\,0\,0\,0\,1\,1\,1\,1\\[-1mm]
1\,0\,1\,0\,1\,0\,1\,0
\end{bmatrix},
\label{eN}
\end{equation}
and we then write:
\begin{equation}
e_{15}\,:~
\begin{bmatrix}
1\,1\,1\,1\,0\,0\,0\,0\,0\,0\,0\,0\,0\,0\,0\\[-1mm]
0\,0\,1\,1\,1\,1\,0\,0\,0\,0\,0\,0\,0\,0\,0\\[-1mm]
0\,0\,0\,0\,1\,1\,1\,1\,0\,0\,0\,0\,0\,0\,0\\[-1mm]
0\,0\,0\,0\,0\,0\,1\,1\,1\,1\,0\,0\,0\,0\,0\\[-1mm]
0\,0\,0\,0\,0\,0\,0\,0\,1\,1\,1\,1\,0\,0\,0\\[-1mm]
0\,0\,0\,0\,0\,0\,0\,0\,0\,0\,1\,1\,1\,1\,0\\[-1mm]
1\,0\,1\,0\,1\,0\,1\,0\,1\,0\,1\,0\,1\,0\,1
\end{bmatrix},\qquad
e_{16}\,:~
\begin{bmatrix}
1\,1\,1\,1\,0\,0\,0\,0\,0\,0\,0\,0\,0\,0\,0\,0\\[-1mm]
0\,0\,1\,1\,1\,1\,0\,0\,0\,0\,0\,0\,0\,0\,0\,0\\[-1mm]
0\,0\,0\,0\,1\,1\,1\,1\,0\,0\,0\,0\,0\,0\,0\,0\\[-1mm]
0\,0\,0\,0\,0\,0\,1\,1\,1\,1\,0\,0\,0\,0\,0\,0\\[-1mm]
0\,0\,0\,0\,0\,0\,0\,0\,1\,1\,1\,1\,0\,0\,0\,0\\[-1mm]
0\,0\,0\,0\,0\,0\,0\,0\,0\,0\,1\,1\,1\,1\,0\,0\\[-1mm]
0\,0\,0\,0\,0\,0\,0\,0\,0\,0\,0\,0\,1\,1\,1\,1\\[-1mm]
1\,0\,1\,0\,1\,0\,1\,0\,1\,0\,1\,0\,1\,0\,1\,0
\end{bmatrix},
\end{equation}
and so on.

These are famous codes: $e_7$ is known as the Hamming $[7,3]$ code, and $e_8$ is the parity-extended Hamming code. Ref.\cite{rCPS}, describes what it calls ``Construction A'', which determines a lattice as a subset of $\ZZ^N$ of all the points whose coordinates modulo $2$ are in the code, and under this, we form the famous lattices $e_7$ and $e_8$.  The points that are of closest distance to the origin form the root lattice for the Lie algebras $E_7$ and $E_8$, respectively. The previously defined $d_N$ codes likewise relate to the root lattice for the Lie algebras $D_N$.

Besides the trivial doubly even code $t_N$, $\{000\cdots 0\}$, for any $N\equiv0\pmod4$, there is an $[N,1]$ doubly even code $h_N$ consisting of $\{000\cdots 0,111\cdots 1\}$, the generating set of which is $\{111\cdots 1\}$.\Ft{This is the only code mentioned here not specifically named in Ref.\cite{rCHVP}.} Note that $h_4=d_4$, but $h_N\subset d_N$ for $N=8,12,16,\dots$.

There are many other doubly even codes, and the number grows quickly as $N$ becomes large; see Appendix~\ref{app:RM} and Refs.\cite{rCPS,rBilRees}.

\subsection{Permutation Equivalence and $R$-Symmetry}
It is also possible to permute the columns in a code.  For instance, for $e_7$ we might swap the last two columns and obtain a generating set
\begin{equation}
\begin{bmatrix}
1\,1\,1\,1\,0\,0\,0\\
0\,0\,1\,1\,1\,0\,1\\
1\,0\,1\,0\,1\,1\,0
\end{bmatrix}.
\end{equation}
This is another doubly even code, and it is different from $e_7$ as given in \Eq{eN}.  To verify this, one could write the 8 codewords in both cases and compare.
More generally, any permutation of columns of a code will produce another code, which is sometimes the same code, sometimes not.

An example where a column-permutation results in precisely the same code again can be seen by taking the $e_7$ generating set\eq{eN}, and swapping the first and third column, then swapping the second and fourth column.  The result would be:
\begin{equation}
\begin{bmatrix}
1\,1\,1\,1\,0\,0\,0\\
0\,0\,1\,1\,1\,1\,0\\
1\,0\,1\,0\,1\,0\,1
\end{bmatrix}
\quad\longrightarrow\quad
\begin{bmatrix}
1\,1\,1\,1\,0\,0\,0\\
1\,1\,0\,0\,1\,1\,0\\
1\,0\,1\,0\,1\,0\,1
\end{bmatrix}.
\end{equation}
The result does not look like the $e_7$ generating set, but it generates the same code.  Indeed, replace the second generator with the sum of the first and the second generator, and we recover exactly the original generator set for $e_7$.

If a permutation of the columns sends one code to another, the codes are said to be {\em\/permutation-equivalent}.  It is convenient for the classification and naming of codes to give one name for the permutation equivalence class, and recognize the multiplicity of codes that the name represents.

Since the columns of a code correspond to the various $Q_I$, a permutation of the columns of the code corresponds to a permutation of the $Q_I$, \ie, to an $R$-symmetry.  For real $N$-extended supersymmetry, the group of $R$-symmetries is $O(N)$; the permutation equivalences describe the subgroup of this matrix group consisting of permutation matrices.  Though this might suggest that the physically relevant question is permutation equivalence of codes, this is not necessarily so: It may well be possible to construct a theory with two types of supermultiplets, corresponding to two different but permutation-equivalent codes, coupled in a way that precludes rewriting the same theory in terms of only one type of supermultiplet. Although different in technical detail, the inextricable coupling of chiral and twisted-chiral supermultiplets discovered in Ref.\cite{rGHR} is a conceptual paradigm of this possibility.

The columns also correspond to the colors of the Adinkra, so permutation equivalence classes give rise to Adinkra topologies (without the edge colors).  This raises the question: do permutation equivalence classes of doubly even codes classify connected Adinkra topologies?  Certainly we have just described a map from the set of permutation equivalence class of doubly even codes to the set of connected Adinkra topologies.  And certainly this map is surjective.  But is it injective?  That is, is it possible that two non-equivalent doubly even codes will give rise to the same Adinkra topology?  The answer to this question is not clear, but luckily, in trying to classify Adinkras, we can leapfrog the issue of classifying Adinkra topologies and instead use the classification of Adinkra chromotopologies, where the issue is clear.

We have created computer algorithms to find and count permutation equivalence classes of doubly even codes.  Table~\ref{t:G2} provides a list of such for $N\le 11$, but beyond this the list is too long to include here.  Instead we provide Table~\ref{t:G3}, which provides the number of permutation equivalence classes of doubly even codes for $N\le 32$, except for a few cases where we do not yet know the answer.  Details on how these tables were created are in Appendix~\ref{app:RM}.
\begin{table}[ht]
\begin{center}
{\small
\begin{tabular}{l|rr|rr|rr|rr|rr|}
\boldmath$N$&\multicolumn{2}{c|}{\boldmath$k=0$}&\multicolumn{2}{c|}{\boldmath$k=1$}
            &\multicolumn{2}{c|}{\boldmath$k=2$}&\multicolumn{2}{c|}{\boldmath$k=3$}
            &\multicolumn{2}{c|}{\boldmath$k=4$}\\\hline
4&$t_4$&(1)&$d_4$&(1)\\\hline
5&$t_5$&(1)&$t_1\oplus d_4$&(5)\\\hline
6&$t_6$&(1)&$t_2\oplus d_4$&(15)&$d_6$&(15)\\\hline
7&$t_7$&(1)&$t_3\oplus d_4$&(35)&$t_1\oplus d_6$&(105)&$e_7$&(30)\\\hline
8&$t_8$&(1)&$t_4\oplus d_4$&(70)&$t_2\oplus d_6$&(420)&$t_1\oplus e_7$&(240)&
$e_8$&(30)\\
&&&$h_8$&(1)&$d_4\oplus d_4$&(35)&$d_8$&(105)&&\\\hline
9&$t_9$&(1)&$t_5\oplus d_4$&(126)&$t_3\oplus d_6$&(1260)&$t_2\oplus e_7$&(1080)&
$t_1\oplus e_8$&(270)\\
&&&$t_1\oplus h_8$&(9)&$t_1\oplus d_4\oplus d_4$&(315)&$t_1\oplus d_8$&(945)&&\\\hline
10&$t_{10}$&(1)&$t_6\oplus d_4$&(210)&$t_4\oplus d_6$&(3150)&$t_3\oplus e_7$&(3600)&$t_2\oplus e_8$&(1350)\\
&&&$t_2\oplus h_8$&(45)&$t_6 * d_6$&(630)&$d_4\oplus d_6$&(3150)&$d_{10}$&(945)\\
&&&&&$t_2\oplus d_4\oplus d_4$&(1575)&$t_2\oplus d_8$&(4725)&&\\\hline
11&$t_{11}$&(1)&$t_7\oplus d_4$&(330)&$t_5\oplus d_6$&(6930)&$t_4\oplus e_7$&(9900)&
$t_3\oplus e_8$&(4950)\\
&&&$t_3\oplus h_8$&(165)&$t_1 \oplus t_6 * d_4$&(6930)&$t_1\oplus d_4\oplus d_6$&(34650)&$t_1\oplus d_{10}$&(10395)\\
&&&&&$t_3\oplus d_4\oplus d_4$&(5775)&$t_3\oplus d_8$&(17325)&$d_4\oplus e_7$&(9900)\\
&&&&&&&$t_5*d_6$&(13860)&&\\\hline
\end{tabular}}
\end{center}
\caption{A listing of permutation equivalence classes for $N$ up to $11$, with the number of codes in the permutation equivalence class given in in the parentheses. Here, $\oplus$ denotes a vector space direct sum, so that if $U \subset (\ZZ_{2})^{N}$ and $V \subset (\ZZ_{2})^{M}$, then $U\oplus V \subset (\ZZ_{2})^{N} \oplus (\ZZ_{2})^{M}\cong (\ZZ_{2})^{N+M}$. The notation $t_{M}*C$ denotes the direct sum together with at least one additional ``glue'' codeword extending into the $t_{M}$ summand, similar to how a $e_{N}$ code is constructed from the corresponding $d_{N}$ code.}
\label{t:G2}
\end{table}
\begin{table}[ht]
\begin{center}\scriptsize
\begin{tabular}{r@{~~}|@{~~}r@{~}|@{~}r@{~}|@{~}r@{~}|@{~}r@{~}|@{~}r@{~}|@{~}r@{~}|@{~}r@{~}|@{~}r@{~}|@{~}r@{~}|@{~}r@{~}|@{~}r@{~}|@{~}r@{~}|@{~}r@{~}|@{~}r@{~}|@{~}r@{~}|@{~}r}
\hbox to0pt{\hss\small\boldmath$_N\!\backslash\!^k$}\kern-2pt
 &\bf1&\bf2&\bf3&\bf4&\bf5&\bf6&\bf7&\bf8&\bf9&\bf10&\bf11&\bf12&\bf13&\bf14&\bf15&\bf16\\[1mm]\hline\hline
 \bf 4 & 1\\\hline
 \bf 5 & 1\\\hline
 \bf 6 & 1& 1\\\hline
 \bf 7 & 1& 1& 1\\\hline
 \bf 8 & 2& 2& 2& 1\\\hline
 \bf 9 & 2& 2& 2& 1\\\hline
 \bf10 & 2& 3& 3& 2\\\hline
 \bf11 & 2& 3& 4& 3\\\hline
 \bf12 & 3& 5& 7& 7& 2\\\hline
 \bf13 & 3& 5& 8& 8& 4\\\hline
 \bf14 & 3& 7& 12& 14& 9& 4\\\hline
 \bf15 & 3& 7& 15& 20& 15& 8& 2\\\hline
 \bf16 & 4& 10& 23& 38& 36& 23& 9& 2\\\hline
 \bf17 & 4& 10& 25& 45& 50& 34& 14& 3\\\hline
 \bf18 & 4& 13& 34& 72& 94& 79& 35& 9\\\hline
 \bf19 & 4& 13& 40& 94& 146& 141& 75& 19\\\hline
 \bf20 & 5& 17& 57& 158& 295& 353& 231& 84& 10\\\hline
 \bf21 & 5& 17& 63& 194& 439& 629& 494& 198& 38\\\hline
 \bf22 & 5& 21& 83& 298& 812& 1481& 1465& 740& 187& 25\\\hline
 \bf23 & 5& 21& 95& 387& 1287& 2970& 3811& 2362& 714& 119& 11\\\hline
 \bf24 & 6& 27& 129& 607& 2444& 7287& 12395& 10048 & 3710 & 739 & 94 & 9\\\hline
 \bf25 & 6& 27& 141& 755& 3808& 15177& 35916 & 38049 & 16039 & 2973 & 309 & 22 \\\hline
 \bf26 & 6& 32& 180& 1114& 6923& 37455& 128270& 194626& 103527& 20206& 1829& 103 \\\hline
 \bf27 & 6& 32& 202& 1435& 11320& 86845 & 464579 & 1103023 & 817167 & 174809& 13578& 525 \\\hline
 \bf28 & 7& 39& 263& 2136& 20812& 224825 & 1917212 & 7631323 & 8948070 & 2550127 & 203178 & 7402& 151 \\\hline
 \bf29 & 7& 39& 287& 2693& 34233& 555804 & 8084014 & * & * & * & 4837471 & 133563 & 1940 \\\hline
 \bf30 & 7& 46& 359& 3866& 61871 & 1477074 & * & * & * & * & * & * & 70744 & 731 \\\hline
 \bf31 & 7& 46& 400& 4972& * & * & * & * & * & * & * & * & * & 25497 & 210 \\\hline
 \bf32 & 8& 55& 506& * & * & * & * & * & * & * & * & * & * & * & 7689 & 85 \\\hline\hline
\hbox to0pt{\hss\small\boldmath$^N\!/\!_k$}\kern-2pt\vrule width0pt height3ex
 &\bf1&\bf2&\bf3&\bf4&\bf5&\bf6&\bf7&\bf8&\bf9&\bf10&\bf11&\bf12&\bf13&\bf14&\bf15&\bf16\\\end{tabular}
\end{center}
\caption{Number of distinct permutation classes of doubly even $[N,k]$ codes. The ``\,*\,'' entry indicates codes that are still being enumerated; see {\small\tt http://www.rlmiller.org/de\_codes/} for up-to-date results, including links to listings of the actual codes.}
\label{t:G3}
\end{table}
 
\section{Supersymmetry and Clifford Algebras}
\label{s:ClifAd}
It has been understood for a long time that the supersymmetry algebra\eq{eSuSy} has a formal similarity with the Clifford algebra generated by the Dirac $\G_I$ matrices\Ft{In \eq{eSuSy}, $I,J=1,\cdots,N$ count {\em fermionic\/} dimensions. Accordingly, the quadratic form, $\d_{IJ}$, occurring in the right-hand side of \Eq{eDirac} is positive definite. This is unlike the {\em typical\/} field theory use of the Clifford/Dirac algebra, where $I,J$ would count {\em bosonic\/} dimensions of spacetime, implying the Lorentzian signature for the quadratic form on the right-hand side of \Eq{eDirac}.  Another point is that traditionally there would be a minus sign on the right hand side; but this difference is equivalent to changing the metric from positive definite to negative definite.  In the language of Ref.~\cite{rLM}, this leads to the Clifford algebra $\Cl(0,N)$ instead of $\Cl(N)=\Cl(N,0)$.}:
\begin{equation}
 \big\{\, \G_I \,,\, \G_J \,\big\} = 2\,\d_{IJ}\,\Ione. \label{eDirac}
\end{equation}
One manifestation of this was the study of the spinning particle by Gates and Rana\cite{rGR1,rGR2}, resulting in the Scalar Supermultiplet (which we herein rename {\em Isoscalar Supermultiplet\/}), defined in terms of the ${\cal GR}(d,N)$ algebra, a form of the Clifford algebra.  We will review their work first, and more explicitly describe how it relates to Clifford algebras; this will motivate the more general construction of relating supermultiplets to a representation of the Clifford algebra.

The main idea relating supermultiplets and Clifford representations is to note that the supersymmetry algebra in one dimension and the Clifford algebra differ only in that in the former, there is a factor of $i$ and a derivative $\ddt$.  So, to turn a supermultiplet into a representation of the Clifford algebra, we can simply forget the factors of $i$ and the derivatives $\ddt$.  This construction is in a sense the same construction quotienting by the ideal $(H-1)$ in Ref.~\cite{r6--1}, where $H=i\ddt$.  This also removes the $\ZZ$-grading afforded by the notion of engineering dimension.  Conversely, to take a representation of the Clifford Algebra, we must re-insert factors of $i$ and $\ddt$.  The factor of $i$ must go on the right-hand side of the transformation rules of the fermions (see Appendix~\ref{app:real}) and if we declare that the $\ddt$ goes in that same place, we end up with a supermultiplet called the Isoscalar supermultiplet.

\subsection{The Isoscalar Supermultiplet and the Clifford Algebra}
The Isoscalar supermultiplet consists of component fields $(\f_1,\dots,\f_m|\j_1,\dots,\j_m)$, written as column vectors:
\begin{equation}
 \Phi=\begin{bmatrix}
               \f_1\\ \vdots\\ \f_m
            \end{bmatrix}
 \qquad\text{and}\qquad
 \Psi=\begin{bmatrix}
               \j_1\\ \vdots\\ \j_m
             \end{bmatrix}.
\end{equation}
and with supersymmetry transformations given by the following Ansatz:
\begin{align}
Q_I\Phi &= \IL_I\Psi,\label{scalarmult1}\\
Q_I\Psi &= i\,\IR_I\,\ddt\Phi,\label{scalarmult2}
\end{align}
where $\IL_I$ and $\IR_I$ are $m\times m$ real matrices to be determined. The supersymmetry algebra\eq{eSuSy} then implies
\begin{eqnarray}
\IL_I\IR_J+\IL_J\IR_I&=&2\,\d_{IJ}\Ione,\label{lr1}\\
\IR_I\IL_J+\IR_J\IL_I&=&2\,\d_{IJ}\Ione,\label{lr2}
\end{eqnarray}
where $\Ione$  is the $m\times m$ identity matrix. The $I=J$ cases of these equations imply that $\IR_I = \IL_I^{-1}$.
\Remk
The additional requirement $\IL_I=\pm \IR_I^T$ was also made in Ref.~\cite{rGR1}, but this is needed only when writing Lagrangians and will play no r\^ole in this paper.

We create, for each $I$, a $2m\times 2m$ real matrix $\G_I$ of the form
\begin{equation}
  \G_I\Defl \left[\begin{array}{c|c}
                      {\bf~0}&\IL_I\\\hline
                       \IR_I &{\bf0}
                  \end{array}\right].
 \label{gammalr}
\end{equation}
Equations\eq{lr1} and\eq{lr2} then imply that these $\G_1,\cdots,\G_N$ indeed satisfy \Eq{eDirac}, the algebra of the Dirac gamma matrices in dimension $N$, whence the notation.

We do not insist that $m$ be minimal here (a minimality criterion was used in Refs.~\cite{rGR1,rGR2}), nor do we adopt any particular convention for the type of the $\G_I$ matrices, except that the entries must be real.

The property of being block off-diagonal corresponds to the existence of a fermion number operator $(-1)^F$ which anti-commutes with the $\G_I$.  If we write our fields listing the bosons first, followed by the fermions, then $(-1)^F$ will be a diagonal matrix of the form
\begin{equation}
(-1)^F=\left[\begin{array}{c|c}
\Ione& {\bf0} \\\hline
{\bf0}& -\Ione\end{array}\right].\label{fermionnumber}
\end{equation}
If we define $\G_0$ to be $(-1)^F$, then $\G_0$ anticommutes with the remaining $\G_I$:
\begin{equation}
 \G_0{}^2=+\Ione, \qquad \{\G_0,\G_I\}=0.
\end{equation}
The formal algebra generated by $\G_0, \G_1, \dots, \G_N$, is then defined by the anticommutation relations:
\begin{equation}
\{\G_I,\G_J\}=2\,\delta_{IJ}\Ione,\qquad \text{for }I,J=0,1,\cdots,N\label{eDirac2}
\end{equation}
and therefore $\{\G_0,\ldots,\G_N\}$ satisfy \eq{eDirac}.  The algebra generated by the $\G_1,\dots,\G_N$ is the {\em Clifford algebra} $\Cl(0,N)$, and the inclusion of $\G_0$ extends $\Cl(0,N)$ into the Clifford algebra $\Cl(0,N{+}1)$. Any set of real matrices representing these $\G_0,\G_1,\dots,\G_N$ such that the above algebra closes is called a {\em real Clifford representation}\cite{rLM}.

Now suppose we are given a real Clifford representation.  We can split the representation into the $+1$ and $-1$ eigenvalues for $\G_0$, and if we choose a basis that respects this splitting, then $\G_0$ will be of the form given above, in \eq{fermionnumber}.  The Ansatz \eqs{scalarmult1}{scalarmult2} then gives the transformation rules for the Isoscalar multiplet corresponding to the Clifford representation.

\section{Constructing Adinkras from Codes}
\label{s:construct}
The main construction in this paper is to take a description of an $N$-dimensional cube $[0,1]^N$ and a doubly even $[N,k]$-code $C$, and construct a real Clifford representation.  The Isoscalar supermultiplet corresponding to this real Clifford representation will then be a supermultiplet for $N$-extended supersymmetry in one dimension, and will have an Adinkra with a chromotopology given by $[0,1]^N/C$.

\subsection{Clifford Supermultiplets and Cubical Adinkras}
We start with the $N$-dimensional cube itself.  To obtain an Adinkra with chromotopology the $N$-cube, we take the Clifford Algebra itself, $\Cl(0,N)$, as a Clifford representation.  The Clifford algebra $\Cl(0,N)$ is a real $2^N$-dimensional vector space, spanned by products of the form $\G_{I_1}\cdots \G_{I_k}$, where $I_1<\dots<I_k$\cite{rLM}.  As a vector space, it splits as a direct sum of two vector spaces: the even and odd parts.  The even (resp.~odd) part is the subspace spanned by products of even (resp.~odd) numbers of $\G_I$ matrices.

The Clifford algebra $\Cl(0,N{+}1)$ acts on $\Cl(0,N)$ in the following way: for every $1\le I\le N$, $\G_I$ acts by multiplication on the left.  The operator $\G_0$ multiplies the even $\G$-monomials by $1$ and the odd $\G$-monomials by $-1$.  It is then easy to see that the Clifford algebra\eq{eDirac2} holds.

We take as a basis for $\Cl(0,N)$ the products $\G_{I_1}\dots \G_{I_k}$ as above.  Specifically, for every vertex of the cube $\vec{x}=(x_1,\dots,x_N)\in\{0,1\}^N$, we define a component field\Ft{Elements of a Clifford algebra can be used to realize the component fields of a supermultiplet, by allowing the element of the Clifford algebra to be a function of the time-like coordinate $\tau$.  This construction was explicitly carried out in equation (67) of \cite{rGLP}.} 
\begin{equation}
e_{\vec{x}}=\G_1{}^{x_1}\cdots \G_N{}^{x_N}.\label{cubedefine}
\end{equation}
We note that $\G_I e_{\vec{x}}=\G_I\G_1{}^{x_1}\dots\G_N{}^{x_N}$ can be transformed into the form\eq{cubedefine}, with perhaps an overall minus sign: we use the anticommutation of the $\G_I$'s and, if $x_I=1$, the fact that $\G_I{}^2=1$. This results in
\begin{equation}
\G_I\cdot e_{\vec{x}} = \pm \G_1{}^{x_1}\cdots \G_{I-1}{}^{x_{I-1}} \,\G_I{}^{1-x_I}\,
\G_{I+1}{}^{x_{I+1}}\cdots\G_N{}^{x_N}.\label{cubeedge}
\end{equation}
The sign is $+1$ if the number of $J$ with $x_J=1$ and $J < I$ is even, and is $-1$ otherwise.  If we define $(-1)^{|\vec{x} < I|}$ to be that sign, and define $\vec{x}\,\boxplus I$ to be the vector $(x_1,\ldots,x_{I-1},1-x_I,x_{I+1},\ldots,x_N)$, then this equation becomes
\begin{equation}
\G_I\cdot e_{\vec{x}} = (-1)^{|\vec{x} < I|} e_{\vec{x}\,\boxplus\,I}.
\end{equation}
We will also need the function $\wt(\vec{x})$, which equals the number of 1's in $\vec{x}$.

Thus, $\Cl(0,N)$ is a representation of the Clifford algebra $\Cl(0,N{+}1)$.  It corresponds to an Isoscalar supermultiplet, by replacing $e_{\vec{x}}\mapsto\phi_{\vec{x}}(\t)$ when the weight of $\vec{x}$ is even, and $e_{\vec{x}}\mapsto\psi_{\vec{x}}(\t)$ when the weight of $\vec{x}$ is odd.  Following the construction of the Isoscalar supermultiplet, we define $\IL_I$ and $\IR_I$ to be
\begin{alignat}{3}
  \IL_I \psi_{\vec{x}} &= (-1)^{|\vec{x} < I|}\phi_{\vec{x}\,\boxplus\,I},\qquad&\qquad
  \IR_I \phi_{\vec{x}} &= (-1)^{|\vec{x} < I|}\psi_{\vec{x}\,\boxplus\,I},
 \label{eCuBaseLR}
\intertext{and thus define}
  Q_I \psi_{\vec{x}} &= (-1)^{|\vec{x} < I|}\phi_{\vec{x}\,\boxplus\,I},\qquad&\qquad
  Q_I \phi_{\vec{x}} &= (-1)^{|\vec{x} < I|}i \ddt\psi_{\vec{x}\,\boxplus\,I}.
 \label{eCuBaseQ}
\end{alignat}

The Adinkra for this is an $N$-dimensional cube.  To see this, we note that the basis elements are labeled by $\vec{x}\in\{0,1\}^N$, the vertices of the $N$-dimensional cube.  The edges colored $I$ connect $\vec{x}$ to $\vec{x}\,\boxplus I$, which changes the $I^\text{th}$ coordinate.  The supermultiplet specified by Eqs.\eq{eCuBaseQ} was called the ``base bosonic Clifford Algebra superfield''  in Ref.~\cite{rA}. To emphasize its chromotopology, we will refer to this as the ``Colored $N$-cube Clifford supermultiplet''.

\subsection{The $N=4$, $D_4$ Projection}
\label{s:d4}
In considering Adinkras that are not cubes, but rather quotients of cubes, it is useful to first consider a few examples.  First, we consider the $N=4$ example obtained by quotienting a four-dimensional cube by identifying antipodal points.  This Adinkra was first described in Ref.~\cite{rA}, where it was identified as the dimensional reduction of the $D=4$ chiral superfield.  Here, we present this example in a way that will motivate the general construction to quotient cubes.

\subsubsection{One $D_4$ Projection}
Consider the Clifford representation $\Cl(0,4)$.  Define the element
\begin{equation}
 g=\G_1\G_2\G_3\G_4.\label{g4}
\end{equation}
Define the two linear transformations $\p_{+}, \p_{-}:\Cl(0,N)\to\Cl(0,N)$ to be
\begin{equation}
 \p_{\pm}(v)=v{\cdot}\frac{1\pm g}{2},\quad v,\p_{\pm}(v)\in\Cl(0,N) . 
 \label{eqn:proj}
\end{equation}

Note that in the definition\eq{eqn:proj}, we multiply $v$ by a factor on its right.  This is important.  It implies that for every $\G_I$ and $v\in \Cl(0,N)$, we have that $\G_I\big(\p_{\pm}(v)\big)=\p_{\pm}\big(\G_I(v)\big)$.

The fact that $g$ is even means that $\p_\pm$ preserves the bosonic and fermionic statistics.  The fact that $g^2=1$ implies that $\p_+^2=\p_+$, $\p_-^2=\p_-$, $\p_+\p_-=\p_-\p_+=0$, and $\p_+ + \p_- = \Ione$.  This means that $\p_+$ and $\p_-$ are a complete set of projection operators, so that $\Cl(0,N)=\img(\p_+)\oplus\img(\p_-)$ as Clifford representations, and $\p_+$ and $\p_-$ project onto their corresponding components.  Neither of these components are zero, since $(1+g)/2$ and $(1-g)/2$ are themselves non-zero elements of $\Cl(0,N)$, which are in $\img(\p_+)$ and $\img(\p_-)$, respectively.

This kind of projection is nothing new: the matrix $g$ is, up to a scalar factor, the matrix known as $\g_5$ in four-dimensional field theory, and the projection $\p_\pm$ corresponds to the familiar projection to chiral spinors: the left- and right-handed halves of the Dirac spinor.

It is also true that $\ker(\p_+)=\img(\p_-)$ and vice-versa, so that we can also describe these representations in terms of constraints: as ``$v$ such that $v{\cdot}(1\mp g)=0$''.  Then we can realize the Clifford representation as a subspace of $\Cl(0,N)$, rather than as a quotient.  This accords with the idea that $\Cl(0,N)$, being a representation for $\Cl(0,N{+}1)$, decomposes as a direct sum into irreducibles.

We take the standard basis for $\Cl(0,4)$ mentioned above, $\{e_{\vec{x}}:\vec{x}\in\{0,1\}^N\}$, and apply $\p_+$ (resp.\ $\p_-$) on it.  The result spans $\img(\p_+)$ (resp. $\img(\p_-)$), but there are duplications (up to sign).  For instance, in $\img(\p_+)$, a vertex $\p_+(e_{\vec{x}})$ and $\p_+(e_{\vec{x}})\,g$ will be identified. Since, up to an overall sign,
\begin{equation}
 \p_+(e_{(x_1,x_2,x_3,x_4)}){\cdot}g ~\propto~\p_+(e_{(1-x_1,1-x_2,1-x_3,1-x_4)}),
\end{equation}
vertices are accordingly identified pairwise.

Constructing the Isoscalar supermultiplet from this, we get the following Adinkra:
\begin{equation}
 \vC{
 \begin{picture}(60,40)(5,-5)
 \put(3,3){\includegraphics[height=27mm]{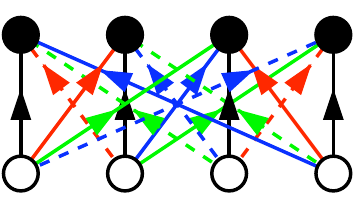}}
 \put(-14,15){$\img(\p_+):$}
 \put(0,30){$\j_{(1000)}$}
 \put(14,30){$\j_{(0100)}$}
 \put(28,30){$\j_{(0010)}$}
 \put(42,30){$\j_{(1110)}$}
 \put(0,0){$\f_{(0000)}$}
 \put(14,0){$\f_{(1100)}$}
 \put(28,0){$\f_{(1010)}$}
 \put(42,0){$\f_{(0110)}$}
 \end{picture}}
 \label{A:D4}
\end{equation}

The result\eq{A:D4} is the four-dimensional cube with opposite corners identified, as in Refs.\cite{rA}.  This Adinkra is a quotient of the four-dimensional cube by the code $d_4$ generated by $1111$.  This corresponds to the fact that by either doing nothing, or by reversing all four bits of a vertex, we return to the same vertex (with perhaps a minus sign).  In reference to the code name, the topology of the four-dimensional cube with opposite corners identified will be called $D_4$.

\subsubsection{The Two Inequivalent $D_4$ Quotients}
Similarly, we can find the Adinkra for the image of $\p_-$, using $1-\G_1\G_2\G_3\G_4$ instead.  When we do so, we see that the image of $\p_-$ and of $\p_+$ look similar---indeed, they have the same topology, the $4$-cube with opposite corners identified.  But they have different patterns of dashed edges.  These patterns cannot be made to coincide even when we redefine some of the vertices by replacing them with their negatives.  There are, in fact, two distinct irreducible representations of $\Cl(0,5)$, and these are the two.  Nevertheless, we can easily describe the relationship between them: by replacing $\G_4\mapsto-\G_4$, for instance, which corresponds to replacing $Q_4$ with $-Q_4$.  This, in turn, results in reversing the sign associated to each edge with $I=4$.  The Adinkras for these two are as follows:
\begin{equation}
 \vC{
 \begin{picture}(120,40)(5,-5)
 \put(3,3){\includegraphics[height=30mm]{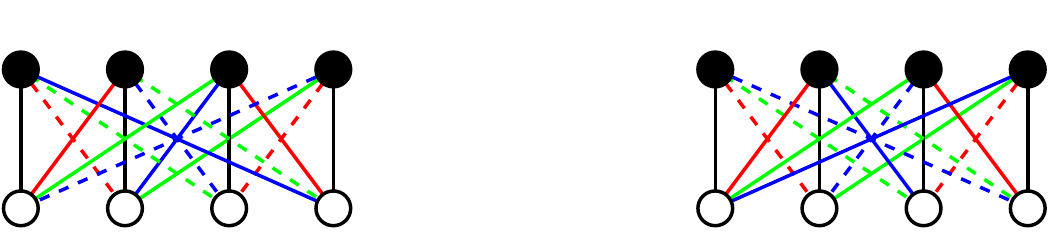}}
 \put(-14,15){$\img(\p_+):$}
 \put(0,30){$\j_{(1000)}$}
 \put(14,30){$\j_{(0100)}$}
 \put(28,30){$\j_{(0010)}$}
 \put(42,30){$\j_{(1110)}$}
 \put(0,0){$\f_{(0000)}$}
 \put(14,0){$\f_{(1100)}$}
 \put(28,0){$\f_{(1010)}$}
 \put(42,0){$\f_{(0110)}$}
 \put(68,15){$\img(\p_-):$}
 \put(82,30){$\j_{(1000)}$}
 \put(96,30){$\j_{(0100)}$}
 \put(110,30){$\j_{(0010)}$}
 \put(124,30){$\j_{(1110)}$}
 \put(83,0){$\f_{(0000)}$}
 \put(97,0){$\f_{(1100)}$}
 \put(111,0){$\f_{(1010)}$}
 \put(125,0){$\f_{(0110)}$}
 \end{picture}}
 \label{eB44a+b}
\end{equation}
We have suppressed the directions of the arrows:
 they are assumed to always point upward.  The
 replacement of $Q_4$ by $- \, Q_4$ is seen in the
 two Adinkras above by noting that all solid blue
 lines in the left hand Adinkra are replaced by
 dashed blue lines in the in the right hand Adinkra
 (and vice-versa).
 
Since the dimensional reduction of the $N=1$ chiral superfield in dimension $D=4$ down to dimension $D=1$ results in the image of $\p_+$, Readers might be tempted to think that the image of $\p_-$ arises from the dimensional reduction of the $N=1$, $D=4$ antichiral superfield.  This, however, is not the case.  Rather, it distinguishes different ways of reconstituting the $Q_1,\cdots,Q_4$ into $Q_\a$ and $Q^\dag_{\dot{\a}}$, with $\a,\dot\a=1,2$, as usual.  To go from one to the other Adinkra\eq{eB44a+b}, we must complex conjugate not both of the components of $Q_\alpha$, but only {\em half\/} of them---which is impossible without violating Lorentz symmetry in four dimensions.  If we complex conjugate all of $Q_\alpha$, we do swap chiral with antichiral superfields, but we reverse {\em\/both\/} $Q_3$ and $Q_4$. This change can be reversed by a redefinition of the real component fields, which swaps their complex combinations into the complex conjugates.
 Graphically, the action of complex conjugation would
 require that {\em {two}} colors must be used to
 implement the dashed/solid exchanges.  So for example,
 both solid blue lines {\em {and}} solid green lines
 in the left hand
 Adinkra are replaced by dashed blue lines  {\em {and}}
 dashed green lines in the in the right hand Adinkra
 (and vice-versa).

 The distinction between the two nonisomorphic supermultiplets depicted by the Adinkras\eq{eB44a+b} is thus more subtle. In fact, in more than 2-dimensional spacetimes, the $Q_I$'s are not Lorentz-invariant, and the sign of only one of them cannot be changed without violating Lorentz symmetry. In 2-dimensional spacetime the supermultiplets depicted in\eq{eB44a+b} are called chiral and twisted-chiral\cite{rGHR}, and we adopt this nomenclature also for the worldline $N=4$ supersymmetry. In $(2,2)$-supersymmetric theories in 2-dimensional spacetime, the transformation between the two supermultiplets\eq{eB44a+b} has been identified\cite{rTwSJG1,rMP1} as the root of mirror symmetry\cite{rMMYau1,rMMYau2,rMMYau3}.
 
 Now, any Lagrangian term involving only one of these types of supermultiplets can just as well be written in terms of only the other type; in this sense they may be regarded as equivalent.
 However, these two supermultiplets may well mix in a Lagrangian, and in a way that prevents rewriting the Lagrangian in terms of only one or the other type of supermultiplet, as has been done in Ref.\cite{rGHR}.
 This feature makes the two representations of supersymmetry, corresponding to two distinct irreducible representations of $\Cl(0,5)$ and depicted by the Adinkras\eq{eB44a+b}, {\em\/usefully distinct\/}.

 Note that such two nonisomorphic irreducible representations of $\Cl(0,N{+}1)$ exist precisely when $N=0\pmod4$, according to Table~\ref{cliffordtable}, discussed in Section~\ref{s:CCR}.

\subsection{Projecting Twice: the $D_6$ Isoscalar Supermultiplet}
 \label{s:D6}
For $N=6$, define the two elements
\begin{eqnarray}
 g_1&=&\G_1\G_2\G_3\G_4,\label{g61}\\
 g_2&=&\G_3\G_4\G_5\G_6.\label{g62}
\end{eqnarray}
As before, $g_1^2=g_2^2=1$.  Also note that $g_1$ and $g_2$ commute, and in fact,
\begin{equation}
 g_1g_2=g_2g_1=-\G_1\G_2\G_5\G_6.
 \label{e:g1g2}
\end{equation}
Analogously to the $D_4$ example, we define the four projection operators
\begin{eqnarray}
 \p_{1\pm}(v)=v{\cdot}\frac{1\pm g_1}{2},\\
 \p_{2\pm}(v)=v{\cdot}\frac{1\pm g_2}{2}
\end{eqnarray}
and note that since $g_1$ and $g_2$ commute, so do $\p_{1+}$ and $\p_{2+}$.  We will now project twice: once using $\p_{1+}$ (or $\p_{1-}$) and then again, using $\p_{2+}$ (or $\p_{2-}$)---a total of four choices.  Using $\p_{1+}$ and $\p_{2+}$ (for instance) produces the Clifford representation $\img(\p_{1+}\circ\p_{2+})$.  In this, and in what follows, we will use $\p_{1+}$ and $\p_{2+}$ for notational definiteness, but it is to be understood that these may be replaced by $\p_{1-}$ or $\p_{2-}$, respectively and independently, {\em mutatis mutandis}.

The composition is
\begin{align}
 (\p_{1+}\circ\p_{2+})(v)
 &= v{\cdot}\frac{1+g_2}{2}{\cdot}\frac{1+g_1}{2}
  =\frac{1}{4}v{\cdot}(1+g_1+g_2+g_2g_1)\\
 &=\frac{1}{4}v{\cdot}(1+\G_1\G_2\G_3\G_4+\G_3\G_4\G_5\G_6-\G_1\G_2\G_5\G_6).
\end{align}
It is straightforward to prove that $\img(\p_{1+}\circ\p_{2+})=\img(\p_{1+})\cap\img(\p_{2+})$ and that this is a Clifford representation.  Writing this as $\ker(\p_{1-})\cap\ker(\p_{2-})$, we can see that for all $v$ in this image, $v=v{\cdot}g_1=v{\cdot}g_2=v{\cdot}g_1g_2$.  In particular, when $g_1$ and $g_2$ commute with $v$, we have that $v=g_1\,v=g_2\,v=g_1g_2\,v$.

Define
 \begin{equation}
e_0 := (\p_{1+}\circ\p_{2+})(1)
 = \frac{1}{4}(1+\G_1\G_2\G_3\G_4+\G_3\G_4\G_5\G_6-\G_1\G_2\G_5\G_6).
\end{equation}
Note that $e_0$ commutes with $g_1$ and $g_2$, and so $e_0 = g_1e_0 = g_2e_0 = g_1g_2e_0$.  We successively apply the various $\G_1, \G_2, \G_3, \G_5$ to $e_0$ on the left and we get a collection of 16 fields that span $\img (\p_{1+}\circ \p_{2+})$, corresponding to vectors of the form $e_{(x_1,x_2,x_3,0,x_5,0)}$.  Applying $\G_4$ will not generate any new vectors because $\G_4 = g_1\G_1\G_2\G_3$, and so for every $e_{\vec{x}}$,
\begin{eqnarray}
\G_4 e_{\vec{x}} &=& g_1\G_1\G_2\G_3\cdot e_{\vec{x}}\\
&=& g_1 \G_1\G_2\G_3\G_1^{x_1}\cdots\G_N^{x_N}\cdot e_0\\
&=& \pm \G_1\G_2\G_3\G_1^{x_1}\cdots\G_N^{x_N} g_1\cdot e_0\\
&=& \pm \G_1\G_2\G_3\G_1^{x_1}\cdots\G_N^{x_N}\cdot e_0\\
&=&\pm \G_1\G_2\G_3\cdot e_{\vec{x}}.
\end{eqnarray}
Similarly, $\G_6$ will not generate new vectors, using $\G_6=g_2\G_3\G_4\G_5$.  

The Isoscalar supermultiplet corresponding to this will have 8 bosons and 8 fermions.  The $e_{\vec{x}}$ vectors correspond to the various $\f_{\vec{x}}$ and $\j_{\vec{x}}$.

If we do this, we get an Adinkra whose topology we call $D_6$, which is a six-dimensional cube projected twice: once according to $g_1=\G_1\G_2\G_3\G_4$ and then according to $g_2=\G_3\G_4\G_5\G_6$:
\begin{equation}
 \vC{
 \begin{picture}(120,45)(15,-5)
 \put(4,1){\includegraphics[height=30mm]{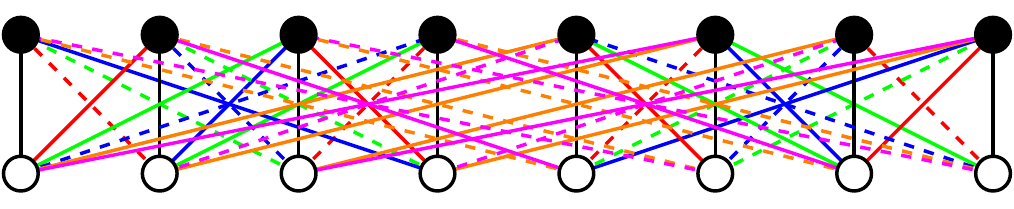}}
 \put(0,32){$\j_{(10000)}$}
 \put(20,32){$\j_{(01000)}$}
 \put(40,32){$\j_{(00100)}$}
 \put(60,32){$\j_{(11100)}$}
 \put(0,0){$\f_{(00000)}$}
 \put(20,0){$\f_{(11000)}$}
 \put(40,0){$\f_{(10100)}$}
 \put(60,0){$\f_{(01100)}$}
 \put(80,32){$\j_{(00001)}$}
 \put(100,32){$\j_{(11001)}$}
 \put(120,32){$\j_{(10101)}$}
 \put(140,32){$\j_{(01101)}$}
 \put(80,0){$\f_{(10001)}$}
 \put(100,0){$\f_{(01001)}$}
 \put(120,0){$\f_{(00101)}$}
 \put(140,0){$\f_{(11101)}$}
 \end{picture}}
 \label{eN6B88}
\end{equation}
Here the colors are as before, with orange for $Q_5$ and purple for $Q_6$.

The name $D_6$ derives from the code $d_6$, generated by the $k=2$ codewords $c_1=111100$ and $c_2=001111$; to denote this, we assemble the generator codewords as a matrix, called a {\em generator matrix} in coding theory:
\begin{equation}
 \left[\begin{smallmatrix}111100\\[3pt] 001111\end{smallmatrix}\right],
\end{equation}
but here we simply regard it as a collection of the codewords listed in its rows.  The first codeword corresponds to $\G_1\G_2\G_3\G_4$, the second with $\G_3\G_4\G_5\G_6$, and the sum of these codewords modulo 2, $110011$, corresponds to $-\G_1\G_2\G_5\G_6$.

More generally, addition in the code turns into multiplication of the corresponding products of $\G_I$ matrices, perhaps with a minus sign, because the $\G_I$ can be anticommuted past each other, and when a particular $\G_I$ appears in both $g_1$ and $g_2$, once anticommutation is done so that the two $\G_I$'s are adjacent, these two simplify to $\G_I{}^2=1$.

\subsection{Constructing an Isoscalar Supermultiplet from a Code}
The procedure illustrated in the previous two examples can be generalized to the following construction, first applied to Clifford representations by A.~Dimakis.\cite{rDimakis}.

\begin{construction}\label{const:quotient}
Suppose we are given $N$ and a doubly-even code $C\subset(\ZZ/2)^N$ of length $N$, given by a generating set $\{c_1,\dots,c_k\}\subset C$. Writing each $c_i$ as $(x_{i1},\ldots,x_{iN})\in\{0,1\}^N$, we associate to it $g_i\Defl\G_1^{x_{i1}}\cdots\G_N^{x_{iN}}$. For instance, $c_1=101100100$ would produce $g_1=\G_1\G_3\G_4\G_7$.

The $\wt(c_i)$ being even translates into $g_i$ being even.  When $\wt(c_i)$ is even, $\wt(c_i)$ being a multiple of $4$ is equivalent to $g_i{}^2=1$.  At the end of Section~\ref{s:Codes} we presented the standard fact that any two elements of the codes are orthogonal (that is, share an even number of $1$s).  As a result, all of the $g_1,\ldots,g_k$ commute with each other.  They generate a group $G$ under multiplication that is isomorphic as a group to the code $C$.  The isomorphism is done analogously to converting $c_i$ to $g_i$, except that a minus sign is sometimes required.\footnote{This construction is due to J. Wood, who used it to classify 2-elementary abelian subgroups of the spin groups\cite{rJAW}.}

For each $g_i$ we have $\p_{i\pm}:\Cl(0,N)\to\Cl(0,N)$ defined as
\begin{equation}
\p_{i\pm}(v)=v{\cdot}\frac{1\pm g_i}{2}.
\end{equation}
As before, the operators $\p_{i\pm}$ are homomorphisms.  The evenness of $g_i$ implies that $\p_{i\pm}$ preserves the fermionic and bosonic statistics.  The fact that $g_i{}^2=1$ implies that $\p_{i+}$ and $\p_{i-}$ are projection operators.

The fact that the $g_i$ all commute implies that the $\p_{i+}$ and $\p_{i-}$ all commute.

The following table summarizes how the properties of the generators of the code relate to the properties of the $\p_{i\pm}$.
\begin{equation}
\begin{tabular}{c|c|c}
\boldmath\bf$c_i$ (Codeword)
                            &\boldmath$g_i$
                                          &\boldmath\bf$\p_{i\pm}$ (Projector)\\[1pt]
 \hline\hline
even weight                 & even        & preserves statistics\\
weight is~ $0\pmod4$   & $g_i{}^2=1$ & projection \\
pairwise orthogonal         & commute     & commute\\\hline
\end{tabular}
\end{equation}

We define\Ft{Strictly speaking, the notation\eqs{epiC}{eImC} should specify for each $g_i$ which sign is being used. For illustrative purposes, herein we only use $\p_{i+}$'s.  This choice will not affect chromotopology, but will affect the dashedness of the edges. Precisely which of these $2^k$ distinct choices in $\p_{1\pm}\circ\cdots\circ\p_{k\pm}$ provide (non-)isomorphic representations akin to\eq{eB44a+b} is a question we defer to a subsequent effort.}
\begin{equation}
 \p_C \Defl \p_{1+}\circ\cdots\circ\p_{k+},
 \label{epiC}
\end{equation}
and then as in the previous example,
\begin{equation}
 \img(\p_C)~=~\img(\p_{1+})\cap\cdots\cap\img(\p_{k+})
           ~=~\ker(\p_{1-})\cap\cdots\cap\ker(\p_{k-})
 \label{eImC}
\end{equation}
is the Clifford representation we want.  If we define
\begin{equation}
 e_0=\p_C(1)=\frac{1+g_1}{2}\cdots\frac{1+g_k}{2}=\frac{1}{2^k}\sum_{g\in G} g,
\end{equation}
and then successively apply the various $\G_I$ on the left to $e_0$, we obtain a set of elements of $\Cl(0,N)$.  Using the Dirac relations (\ref{eDirac}) shows that many of these are the same, up to an overall sign.  Furthermore, for all $g_i$, we have $g_i e_0 = e_0$, and more generally this is true of all elements of $G$.  It therefore follows that if we begin with an element $v$ represented by a $d_v$-dimensional matrix, the application of $\p_C$ results in a quantity representable by a $2^{-k}d_v$-dimensional matrix.

In fact, this will result in $2^{N-k+1}$ different elements of $\img(\p_C)$, occurring in $\pm$ pairs.  If we arbitrarily choose one from each $\pm$ pair, the result is a collection of $2^{N-k}$ vectors that form a basis for $\img(\p_C)$.
\end{construction}

\subsection{Other Examples} 

\begin{example}
The case $N=8$ allows the code $e_8$, which has the following generator set:
\begin{equation}
 e_8 \gen \left[\begin{smallmatrix}
                   11110000\\[2pt]
                   00111100\\[2pt]
                   00001111\\[2pt]
                   10101010
                \end{smallmatrix}\right].
\end{equation}
The corresponding $g_i$ are:
\begin{eqnarray}
 g_1&=&\G_1\G_2\G_3\G_4,\\
 g_2&=&\G_3\G_4\G_5\G_6,\\
 g_3&=&\G_5\G_6\G_7\G_8,\\
 g_4&=&\G_1\G_3\G_5\G_7.
\end{eqnarray}
These produce the following group elements:
\begin{subequations}
\begin{align}
 g_1g_2&=-\G_3\G_4\G_5\G_6,&
 g_1g_3&=\4\G_1\G_2\G_3\G_4\G_5\G_6\G_7\G_8,\\
 g_1g_4&=\4\G_2\G_4\G_5\G_7,&
 g_2g_3&=-\G_3\G_4\G_7\G_8,\\
 g_2g_4&=\4\G_1\G_4\G_6\G_7,&
 g_3g_4&=\4\G_1\G_3\G_6\G_8,\\
 g_1g_2g_3&=\4\G_3\G_4\G_7\G_8,&
 g_1g_2g_4&=-\G_1\G_4\G_6\G_7,\\
 g_1g_3g_4&=\4\G_2\G_4\G_6\G_8,&
 g_2g_3g_4&=-\G_1\G_3\G_4\G_8,\\
 g_1g_2g_3g_4&=\4\G_1\G_4\G_5\G_8.
\end{align}
\end{subequations}
Recall that each $g$ must be a product of a doubly even number of $\G_I$'s, explaining why are we only now seeing $g$'s with differing numbers of $\G_I$'s.

The Adinkra is below.  This has the feature that every boson is connected to every fermion.  So this is a $K(8,8)$ graph, and is denoted $E_8$.  This was introduced in Refs.\cite{rGR0,rGLPR} as the $N=8$ spinning particle in relation to what was described as a ``supergravity surprise''.  The new color, brown, corresponds to $\G_8$.
\begin{equation}
 \vC{
 \begin{picture}(120,45)(15,-5)
 \put(4,1){\includegraphics[height=30mm]{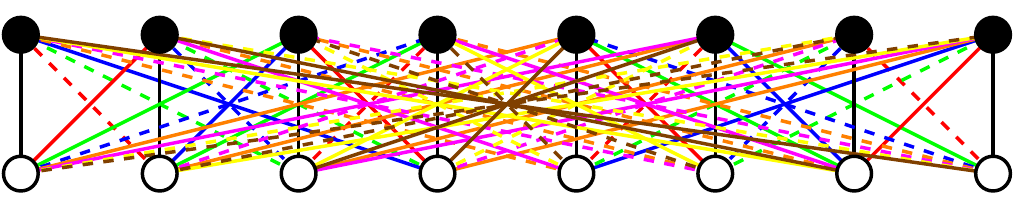}}
 \put(0,32){$\j_{(10000)}$}
 \put(20,32){$\j_{(01000)}$}
 \put(40,32){$\j_{(00100)}$}
 \put(60,32){$\j_{(11100)}$}
 \put(0,0){$\f_{(00000)}$}
 \put(20,0){$\f_{(11000)}$}
 \put(40,0){$\f_{(10100)}$}
 \put(60,0){$\f_{(01100)}$}
 \put(80,32){$\j_{(00001)}$}
 \put(100,32){$\j_{(11001)}$}
 \put(120,32){$\j_{(10101)}$}
 \put(140,32){$\j_{(01101)}$}
 \put(80,0){$\f_{(10001)}$}
 \put(100,0){$\f_{(01001)}$}
 \put(120,0){$\f_{(00101)}$}
 \put(140,0){$\f_{(11101)}$}
 \end{picture}}
 \label{eN8B88}
\end{equation}

As in the $N=4$ case, there are two irreducible representations, and one is obtained as above, while the other is obtained by reversing the sign on one of the $g_i$.  The result is the same Adinkra topology, but with various signs on the edges reversed.  For instance, we can reverse the sign on $g_1$ by reversing the signs on edges corresponding to $\G_4$.  This preserves the sign on $g_2$, $g_3$, and $g_4$.  Reversing various $\G_I$ yields apparently different Adinkras, but all of these must fall into just two isomorphism classes.

There are $\frac{8!}{1344}=30$ codes that are permutation equivalent to $e_8$.

As before, we can take a subgroup of $e_8$, but in this case, the situation is a bit more interesting: there are multiple inequivalent choices for which generator to remove.  Figure~\ref{fig:e8tree} shows all the $N=8$ doubly even codes, up to permutation equivalence.

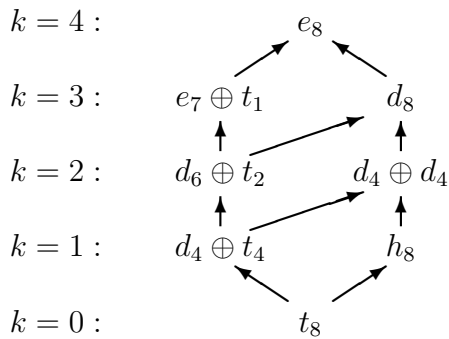
\begin{figure}[htb]
\begin{center}
\begin{picture}(70,45)(0,5)
\put(10,45){$k=4:$}
 \put(50,45){\makebox[0in][c]{$e_8$}}
  \put(40,39){\vector(3,2){7}}
  \put(60,39){\vector(-3,2){7}}
\put(10,35){$k=3:$}
 \put(62,35){\makebox[0in][c]{$d_8$}}
 \put(38,35){\makebox[0in][c]{$e_7\oplus t_1$}}
  \put(38,29){\vector(0,1){4}}
  \put(62,29){\vector(0,1){4}}
  \put(42,28.5){\vector(3,1){15}}
\put(10,25){$k=2:$}
 \put(38,25){\makebox[0in][c]{$d_6\oplus t_2$}}
 \put(62,25){\makebox[0in][c]{$d_4\oplus d_4$}}
  \put(38,19){\vector(0,1){4}}
  \put(62,19){\vector(0,1){4}}
  \put(42,18.5){\vector(3,1){15}}
\put(10,15){$k=1:$}
 \put(38,15){\makebox[0in][c]{$d_4\oplus t_4$}}
 \put(62,15){\makebox[0in][c]{$h_8$}}
  \put(53,9){\vector(3,2){7}}
  \put(47,9){\vector(-3,2){7}}
\put(10,5){$k=0:$}
 \put(50,5){\makebox[0in][c]{$t_8$}}
\end{picture}
\caption{$N=8$ doubly even codes and their subset relationships: The arrows connecting two codes are injections (possibly after permutation): the code on the lower level is a subcode of the higher one.  The trivial code $t_8$ is at the bottom, and $e_8$ is the unique maximal doubly even code for $N=8$ and is drawn at the top.  Every doubly even code in $N=8$ is a subcode of $e_8$.  The code $h_8$ is the one generated by $11111111$.}\label{fig:e8tree}
\end{center}
\end{figure}

For instance, for $k=3$, we could choose $e_7\oplus t_1$, which results in an Adinkra topology $E_7\times I^1$, or we could choose $d_8$, which results in a different Adinkra topology $D_8$:
\begin{align}
 E_7\times I^1&: \vC{\includegraphics[width=140mm]{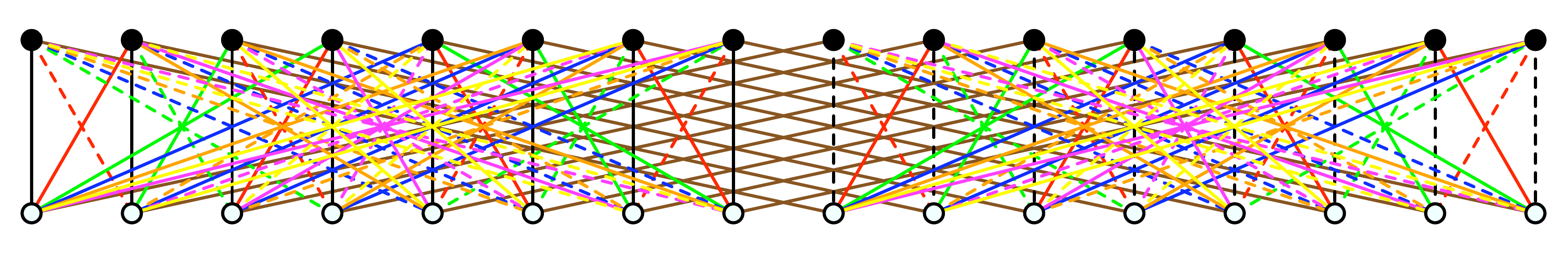}} \label{eE7xI}\\[2mm]
 D_8&:         \vC{\includegraphics[width=140mm]{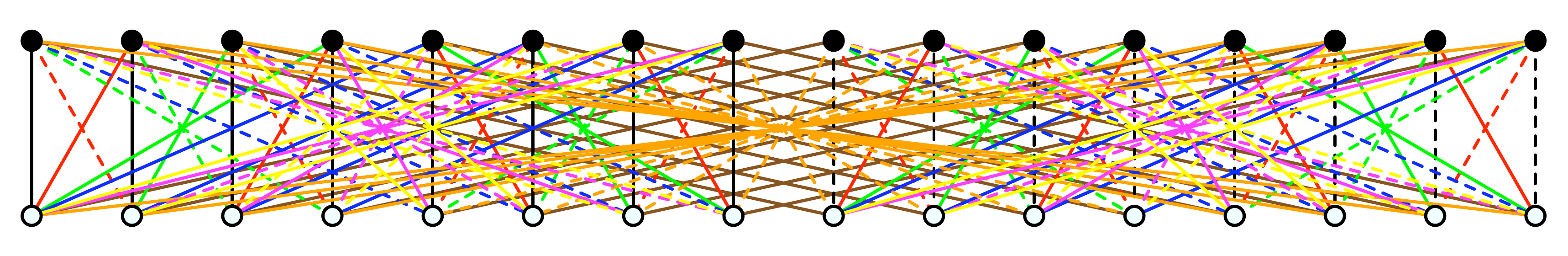}} \label{eD8}
\end{align}
The only difference between these two Adinkras is the way one of the supersymmetries acts, the one represented here by orange edges. In the $E_7\times I^1$-Adinkra\eq{eE7xI}, it acts within each of the two halves, leaving the Adinkra {\em\/1-color-decomposable\/}: only the brown edges span the whole Adinkra, and it decomposes into two identical $N=7$ Adinkras if the brown edges are erased. Each of the halves has the $E_7$ topology, except that edges of one color (black) have their dashing reversed. This becomes clearer upon rearranging the nodes a little:
\begin{equation}
 E_7\times I^1:~ \vC{\includegraphics[width=140mm]{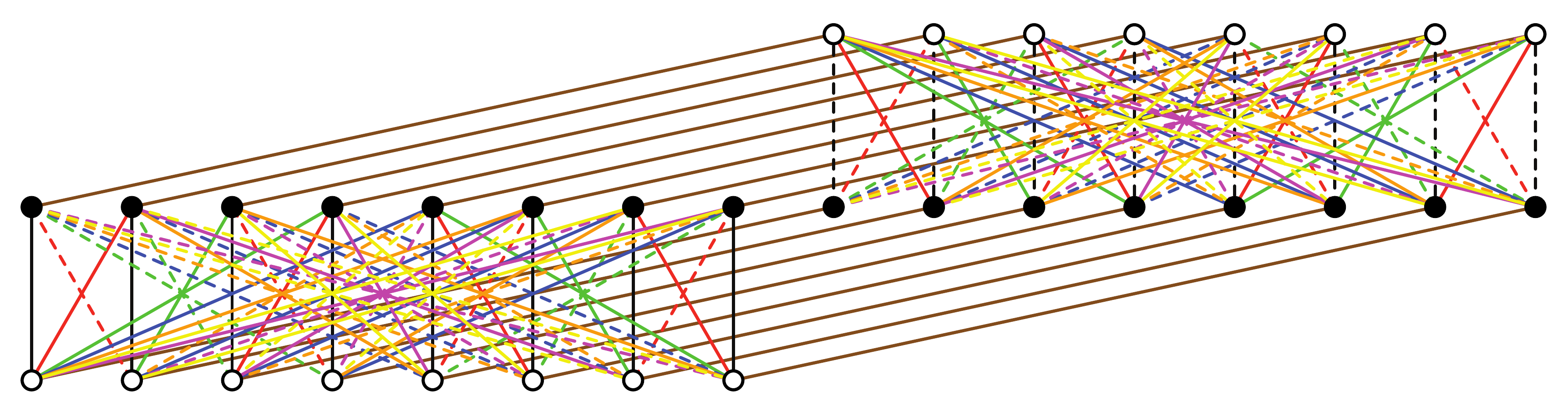}}
 \label{eE7xIa}
\end{equation}

In turn, in the $D_8$-Adinkra\eq{eD8}, the orange edges span the whole Adinkra together with the brown ones. For the sake of comparison with the $E_7\times I^1$-Adinkra\eq{eE7xIa}, we also rearrange the nodes of the $D_8$-Adinkra:
\begin{equation}
 D_8:~ \vC{\includegraphics[width=140mm]{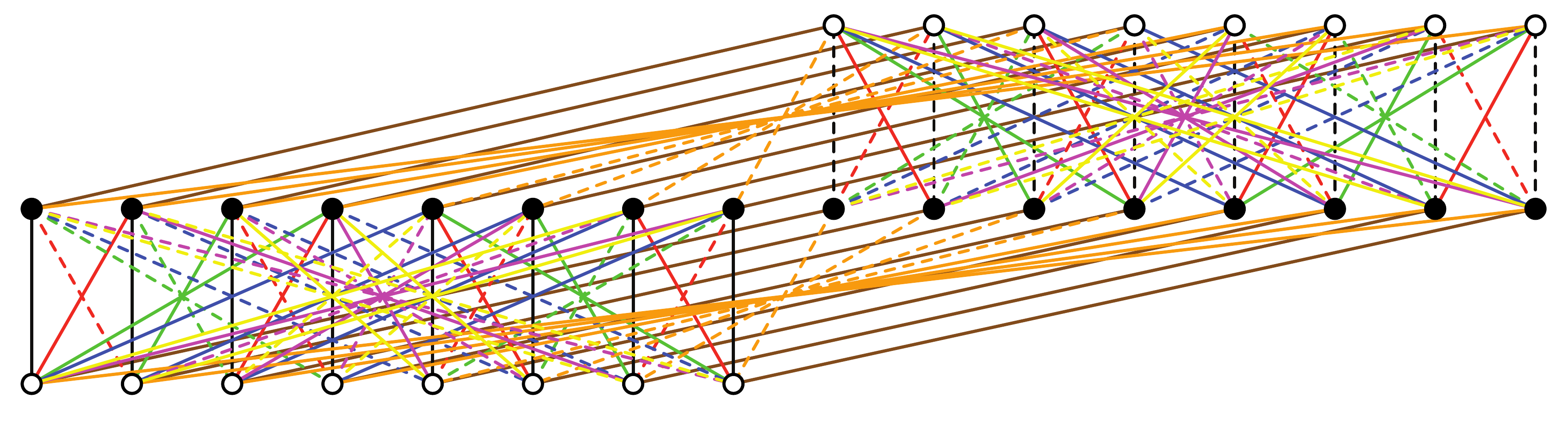}}
 \label{eD8a}
\end{equation}
To decompose this second Adinkra, one would have to erase the edges of at least two colors; we say it is 2-color-decomposable. This $n$-color-decomposability property correlates with the fact that the $e_7\oplus t_1$ code has one column of zeros, whereas $d_8$ does not:
\begin{equation}
 E_7\times I^1 \iff e_7\oplus t_1
  \gen\left[\begin{smallmatrix}
              0000\,1111\\[2pt]
              0011\,1100\\[2pt]
              0101\,0101
            \end{smallmatrix}\right],
 \qquad\text{{\it vs.\/}}\qquad
 D_8 \iff d_8
  \gen\left[\begin{smallmatrix}
              0000\,1111\\[2pt]
              0011\,1100\\[2pt]
              1111\,0000
            \end{smallmatrix}\right].
\end{equation}
Finally, both Adinkras\eq{eE7xI} and\eq{eD8} admit precisely one supersymmetry-preserving $\ZZ_2$ symmetry: for the former, it is encoded as $11110000$ and generated by $\G_1\G_2\G_3\G_4$, while the latter is symmetric with respect to the action of $\G_2\G_4\G_6\G_8$, encoded as $01010101$. The projection of each Adinkra by its respective symmetry then produces the $E_8$-Adinkra\eq{eN8B88}.

Similar relations exist between the $D_6\times I^2$- \textit{vs}.\ $D_4\times D_4$-Adinkras, and the $D_4\times I^4$- \textit{vs}.\ $H_8$-Adinkra corresponding to the $k=2$ and $k=1$ rows in the diagram in Figure~\ref{fig:e8tree}.
\end{example}

To summarize, Table~\ref{topologytable} gives the possible topologies for each $N$ up to $N=10$.  The cartesian product $\times$ refers to taking the cartesian product of the vertex set, and drawing edges between $(v_1,w_1)$ and $(v_1,w_2)$ whenever there is an edge between $w_1$ and $w_2$, and between $(v_1,w_1)$ and $(v_2,w_1)$ whenever there is an edge between $v_1$ and $v_2$.  Exponentiation means iterated cartesian products.  We note the Adinkra topologies $D_4$, $D_6$, $D_8$, and $D_{10}$ from the codes $d_4$, $d_6$, $d_8$ and $d_{10}$, as well as $E_7$ and $E_8$ from $e_7$ and $e_8$, respectively.  In addition, there turns up a code with no standard name, generated by $\{1111000000,0011111111\}$, which we tentatively call $D_4*_1\!I^6$.

\begin{table}[ht]
\begin{center}\footnotesize
\begin{tabular}{c|rc|rc|rc|rc|rc}
\boldmath$N$ & \multicolumn{2}{c|}{\boldmath$k=0$} & \multicolumn{2}{c|}{\boldmath$k=1$}
 & \multicolumn{2}{c|}{\boldmath$k=2$} & \multicolumn{2}{c|}{\boldmath$k=3$}
  & \multicolumn{2}{c}{\boldmath$k=4$}\\[1pt]
 \hline\hline
1&$I$   & $(1|1)$   &&&&&&&&\\
2&$I^2$ & $(2|2)$   &&&&&&&&\\
3&$I^3$ & $(4|4)$   &&&&&&&&\\
4&$I^4$ & $(8|8)$   & $D_4$             &$(4|4)$   &&&&&&\\[1mm]
5&$I^5$ & $(16|16)$ & $D_4\times I^1$     &$(8|8)$   &&&&&&\\
6&$I^6$ & $(32|32)$ & $D_4\times I^2$   &$(16|16)$ & $D_6$           & $(8|8)$&&&&\\
7&$I^7$ & $(64|64)$ & $D_4\times I^3$   &$(32|32)$ & $D_6\times I^1$   & $(16|16)$& $E_7$ & $(8|8)$&&\\
8&$I^8$ & $(128|128)$ & $D_4\times I^4$ & $(64|64)$& $D_6\times I^2$ & $(32|32)$& $E_7\times I^1$ & $(16|16)$& $E_8$ & $(8|8)$\\
 &      &             & $H_8$           & $(64|64)$& $D_4\times D_4$ & $(32|32)$ & $D_8$   & $(16|16)$       &  &    \\[1mm]
9&$I^9$ & $(256|256)$ & $D_4\times I^5$ & $(128|128)$ & $D_6\times I^3$ & $(64|64)$& $E_7\times I^2$ & $(32|32)$& $E_8\times I^1$ & $(16|16)$\\
 &      &             & $H_8\times I^1$   & $(128|128)$ & $D_4\times D_4\times I^1$ & $(64|64)$ & $D_8\times I^1$ & $(32|32)$&&\\
10&$I^{10}$ & $(512|512)$& $D_4\times I^6$& $(256|256)$ & $D_6\times I^4$ & $(128|128)$ & $E_7\times I^3$&$(64|64)$& $E_8\times I^2$&$(32|32)$\\
 &          &            & $H_8\times I^1$  & $(256|256)$ & $D_4\times D_4\times I^2$ & $(128|128)$ & $D_8\times I^2$ &$(64|64)$  & $D_{10}$ &$(32|32)$    \\
 &     &                 &                &             & $(D_4*_1\!I^6)^\dagger$   & $(128|128)$ & $D_4\times D_6$ &$(64|64)$ & &\\[1mm]\hline\noalign{\vglue2mm}
 \multicolumn{11}{l}{\parbox{172mm}{\baselineskip=11pt plus1pt $^\dagger$\,Herein, the code name ``$D_4*_1\!I^6$'' denotes that the $D_4$ code is padded by six zeroes and augmented by one additional, ``glue'' generator spanning the positions of the six added zeroes and a sufficient number (here, two) of 1's within $D_4$ so as to generate a doubly even code; it is generated by $\{1111\,0000\,00,0011\,1111\,11\}$.}}\\[-2mm]
\end{tabular}
\end{center}
\caption{Adinkra Topologies up to $N=10$: The number of nodes in an Adinkra with the indicated topology is shown to the right of each topology, in the form $(n_B|n_F)$, where $n_B$ is the number of bosons and $n_F$ is the number of fermions in the supermultiplet.}
\label{topologytable}
\end{table}

\subsection{Comparison with Clifford Representations}
 \label{s:CCR}
Since Clifford representations are already classified, it is worthwhile comparing what we have just found with the known classification of Clifford representations.  The abstract algebras $\Cl(0,N{+}1)$ are known for all $N$.  The following can be found on any standard text on Clifford algebras, such as Lawson and Michelsohn's {\em Spin Geometry}\cite{rLM}.

The Clifford algebras $\Cl(0,N)$ and $\Cl(0,N{+}1)$ are given in Table~\ref{cliffordtable}.  The description of these algebras has a modulo $8$ sort of periodicity, so it is convenient to write $N=8m+s$ where $m$ and $s$ are integers and $0\le s\le 7$.  The notation $\IR(n)$, $\IC(n)$, and $\IH(n)$ denotes the algebra of $n\times n$ matrices with real, complex, and quaternionic coefficients, respectively.

\begin{table}[ht]
\begin{center}
{\scriptsize
\begin{tabular}{r|c|c|c|c|c|c|c}
\boldmath$s$ & \bf\boldmath Cl$(0,N)$ & \bf\boldmath Cl$(0,N{+}1)$ & \bf\boldmath$\dim_\IR\text{Irrep.}$
 & \bf\boldmath$\dim_\IR\text{Cl}(0,N)$ &
 \parbox[b]{25mm}{\centering\baselineskip=10pt\bf\boldmath\#(Irrep.) in $\text{Cl}(0,N)$}
 & \bf\boldmath$\max k$
 \\[1pt]\hline\hline
0&$\IR(16^m)$&$\IR(16^m)\oplus\IR(16^m)$&$16^m$&$16^{2m}$&$16^m$&$4m$\\
1&$\IR(16^m)\oplus\IR(16^m)$&$\IR(2\cdot 16^m)$&$2\cdot 16^m$&$2\cdot 16^{2m}$&$16^m$&$4m$\\
2&$\IR(2\cdot 16^m)$&$\IC(2\cdot 16^m)$&$4\cdot 16^m$&$4\cdot 16^{2m}$&$16^m$&$4m$\\
3&$\IC(2\cdot 16^m)$&$\IH(2\cdot 16^m)$&$8\cdot 16^m$&$8\cdot 16^{2m}$&$16^m$&$4m$\\
4&$\IH(2\cdot 16^m)$&$\IH(2\cdot 16^m)\oplus\IH(2\cdot 16^m)$&$8\cdot 16^m$&$16\cdot 16^{2m}$&$2\cdot 16^m$&$4m+1$\\
5&$\IH(2\cdot 16^m)\oplus\IH(2\cdot 16^m)$&$\IH(4\cdot 16^m)$&$16\cdot 16^m$&$32\cdot 16^{2m}$&$2\cdot 16^m$&$4m+1$\\
6&$\IH(4\cdot 16^m)$&$\IC(8\cdot 16^m)$&$16\cdot 16^m$&$64\cdot 16^{2m}$&$4\cdot 16^m$&$4m+2$\\
7&$\IC(8\cdot 16^m)$&$\IR(16\cdot 16^m)$&$16\cdot 16^m$&$128\cdot 16^{2m}$&$8\cdot 16^m$&$4m+3$\\\hline
\end{tabular}\vspace{-3mm}
}
\end{center}
\caption{Clifford algebras $\Cl(0,N)$ decomposed into irreducible representations of $\Cl(0,N{+}1)$:
 Here $N=8m+s$ where $m$ and $s$ are integers and $0\le s\le 7$.}
\label{cliffordtable}
\end{table}

Note from Table~\ref{cliffordtable} that every Clifford algebra is either a matrix algebra, or a direct sum of two matrix algebras.  It is a classical result that every finite-dimensional real representation of such an algebra decomposes into irreducibles.  For $\IR(n)$, $\IC(n)$, and $\IH(n)$, there is up to isomorphism only one irreducible real representation: $\IR^n$, $\IC^n$, or $\IH^n$, respectively.  For $\IR(n)\oplus\IR(n)$, there are two non-isomorphic irreducible representations: one that ignores the second summand and is the standard representation on the first, and the other that ignores the first summand.  Likewise for $\IC(n)\oplus\IC(n)$ and $\IH(n)\oplus\IH(n)$.

The fourth column in Table~\ref{cliffordtable} counts the real dimension of the irreducible representation.  Based on the real dimension of $\Cl(0,N)$, this determines the number of copies of the irreducible representation found in $\Cl(0,N)$, which is shown in column six.

The last column of the table is labeled ``max $k$'', and refers to the maximal dimension $k$ of a doubly-even code.  This can be computed using Gaborit's mass formula, described in Appendix~\ref{sec:gaborit}\cite{rPGMass}, which gives the number of doubly even codes for a given $N$ and $k$.  This number vanishes when $k$ is larger than the numbers given in the table.  To find this value of $k$, recall that in \eq{eKmax} we defined the following function of $N$:
\begin{equation}
 \vk(N):=
  \begin{cases}
   0 &\text{for $N<4$},\\
   1 &\text{for $N=4,5$}, \\
   2 &\text{for $N=6$}, \\
   3 &\text{for $N=7$}, \\
   4 + \vk(N{-}8) &\text{for $N\geq 8$, recursively}.
  \end{cases}
\end{equation}
Note that if we write $N=8m+s$ where $0\le s\le 7$, then
\begin{equation}
 \vk(N)=
  \begin{cases}
   4m &\text{for $s=0$, $1$, $2$, or $3$},\\
   4m+1 &\text{for $s=4$ or $s=5$}, \\
   4m+2 &\text{for $s=6$}, \\
   4m+3 &\text{for $s=7$}.
  \end{cases}
  \label{eq:kmax}
\end{equation}
and Gaborit's mass formula implies that for any doubly even code of length $N$ and dimension $k$, $0\le k\le \vk(N)$.

Suppose we take an $N$-dimensional cubical Adinkra.  It has $2^N$ nodes.  If we quotient by a doubly even code of dimension $k$, so that the code has $2^k$ elements, we have $2^{N-k}$ nodes in the quotient Adinkra.  The minimal Adinkra then occurs when $k$ is as large as possible; that is, for $k=\vk(N)$.  Thus, a minimal Adinkra for a given $N$ has $2^{N-\vk(N)}$ nodes: half bosons, and the other half of them fermions.  In turn, this gives rise to a representation of the Clifford algebra, which in general will be some number of copies of the irreducible representations.  The fact that the real dimension of the irreducible representation, found in column 4, coincides with the real dimension of the maximally quotiented Adinkra, $2^{N-\vk(N)}$, indicates that the irreducible representations of the Clifford algebra can actually be obtained by quotienting by maximal codes.  In particular, the corresponding isoscalar multiplet has a description in terms of Adinkras.

Incidentally, it may concern the reader that this table shows that there is a unique irreducible representation for the Clifford algebra when $N$ is not a multiple of 4, and that there are precisely two irreducible representations for the Clifford algebra when $N$ is a multiple of 4.  This seems strange in contrast to the multitude of codes of maximal $k$ in Table~\ref{topologytable}.  This indicates that for Isoscalar Adinkras at least, there are hidden isomorphisms between these supermultiplets.  The one-hooked versions, however, have no isomorphisms, and so are immune to this problem.  This issue will be investigated in greater detail in a future work.

\subsection{One-Hook Hanging Adinkras}
Since we have concluded that any connected Adinkra chromotopology is a quotient of the $N$-cube by a doubly even $[N,k]$-code $C$, the vertices of the Adinkra (that is, the component fields of the supermultiplet) correspond to cosets of $\{0,1\}^N$, thought of as $(\ZZ_2)^N$, by the subgroup $C$. This immediately implies that the Adinkra has $2^{N-k}$ nodes, where $k$ is the dimension of $C$. Of these, one half represents bosonic component fields, and the other half fermionic component fields in the corresponding supermultiplet.  We thus have that the number of bosonic component fields, $d_B$, and the number of fermionic component fields, $d_F$, after taking this quotient satisfy 
\begin{equation}
d_B ~=~ d_F ~=~  2^{N - k - 1}.
 \label{eDBF1}
\end{equation}

In Ref.\cite{r6-1}, we described a notion of hanging a graph by a one or more sources.  For instance, if we pick the vertex $v_*$ and hang the graph by it, we let all the arrows on edges point from the vertices that are further away from $v_*$ (as measured through the edge set) to vertices that are closer to $v_*$.  Equivalently, for each vertex $v$ we define the engineering degree of $v$ to be
\begin{equation}
 [v]=[v_*]-\frac{1}{2}\mbox{dist}(v,v_*),
 \label{v*v}
\end{equation}
where $\mbox{dist}(v,v_*)$ is the length of the shortest path from $v$ to $v_*$ in the edge set.  Then the arrows are drawn in the direction of increasing engineering degree, upward. 

For every doubly even code $C$, use the Construction above to find an Adinkra.  Then choose one bosonic field and create a one-hooked Adinkra hooked on that field.  Call this Adinkra $\cA_C$.  In this way, we have assigned to every doubly even code $C$ an Adinkra $\cA_C$.  Now given two different doubly even codes $C_1$ and $C_2$, we have two distinct Adinkras, $\cA_1$ and $\cA_2$.  They are not isomorphic as chromotopologies because it is a different set of sequences of colors that take us from a vertex to itself in both cases.  But because these are one-hooked, we can say something stronger: the supermultiplets they describe will not be isomorphic. 

\begin{theorem}
Let $C_1$ and $C_2$ be different doubly even codes, and suppose $\cA_1$ and $\cA_2$ are one-hooked connected Adinkras with codes $C_1$ and $C_2$, respectively.  Then the supermultiplets from $\cA_1$ and from $\cA_2$ are not isomorphic.\label{thm:unique}
\end{theorem}
\begin{proof}
 For suppose there were such an isomorphism.  Suppose $\phi_1$ is the hooked bosonic field for $\cA_1$ and $\phi_2$ is the hooked bosonic field for $\cA_2$.  Any isomorphism preserves the number of degrees of freedom in each mass dimension, and since there are no fields with mass dimension less than $[\phi_1]$ or $[\phi_2]$, we must have $[\phi_1]=[\phi_2]$.  Any isomorphism must send some linear combination of component fields of $\cA_1$ and their derivatives to $\phi_2$.  But $\phi_2$ is of lowest mass dimension, and derivatives increase mass dimension.  So the only candidate for what must be sent to $\phi_2$ is $\phi_1$ perhaps multiplied by a unit-less constant $c$.
 
The fact that an isomorphism commutes with $Q_I$ means that for each $I$, $c Q_I\phi_1$ gets sent to $Q_I \phi_2$.  We continue in this fashion, until the isomorphism of supermultiplets produces an isomorphism of the corresponding chromotopologies.

By Proposition~\ref{prop:codeinvariant}, the code is an invariant of the chromotopoogy, and therefore, this isomorphism of chromotopologies implies $C_1=C_2$.
\end{proof}

One final remark about one-hooked Adinkras: for such Adinkras, we can examine how many component fields are in each engineering degree.  By \Eq{v*v}, this is equivalent to finding how many vertices are of a given distance from the highest vertex, $v_*$.  The corresponding notion in coding theory is the ``coset weight enumerator''\cite{rCHVP}.  This is a polynomial of the form $\sum_\ell a_\ell\, x^\ell$ where $a_\ell$ is the number of cosets whose distance to the $0$ coset is $\ell$.  Thus, the coset weight enumerator for a doubly even code can be used to find the number of component fields in each engineering dimension for a one-hooked Adinkra corresponding to that doubly even code.

\section{Application to 4- and Higher-Dimensional Theories}\label{s:fourd}
We expect to find, among $D=1$ off-shell theories, the dimensional reductions of the off-shell theories of higher-dimensional supersymmetric theories.  For instance, off-shell theories in four dimensions with $\cal N$-extended supersymmetry will dimensionally reduce to 1-dimensional theories with $N=4{\cal N}$ supersymmetries, since the smallest irreducible spinor in 4-dimensional spacetime has $N=4$ components.  In higher dimensions the smallest irreducible spinors have dimension that is an integral multiple of this.

Multiplets with minimal numbers of degrees of freedom are of interest, not only because they presumably provide the simplest theories, but also because they give us an idea of a lower bound to the size of a theory.  Minimal multiplets correspond to the irreducible representations of the Clifford Algebra, which are described in Table~\ref{cliffordtable}.  As was pointed out in Section~\ref{s:CCR}, they also correspond to maximal codes for each $N$, and these are listed in Table~\ref{t:G3} for $N\le 32$.  In any case, the number of component fields is $2^{N-\vk(N)}$ so that the number of bosons is $2^{N-\vk(N)-1}$ and likewise for the number of fermions.

The results are shown in Table~\ref{4DminREPsiz}.

\begin{table}[ht]
\begin{center}
\begin{tabular}{c|r|c|c|rl} 
 \boldmath$\cal N$ & \boldmath$N$ & \bf\#(Bosons) & \bf\#(Fermions) &
  \multicolumn{2}{c}{\bf \#(Adinkra Topologies)}\\[1pt]
 \hline\hline
$1$ & 4  & 4 & 4 & 1 & $D_4$  \\ 
$2$ & 8  & 8 & 8  &  1 & $E_8$ \\ 
$3$ & 12 & 64 & 64  & 2 & $D_{12}$, $E_8\times D_4$  \\ 
$4$ & 16 & 128 & 128  &  2 & $E_{16}$, $E_8\times E_8$ \\[1mm]
$5$ & 20 & 1,024 & 1,024  &  10 & $\dots^*$\\ 
$6$ & 24 & 2,048 & 2,048  & 9 & $\dots^*$\\ 
$7$ & 28 & 16,384 & 16,384  &  151 & $\dots^*$\\ 
$8$ & 32 & 32,768 & 32,768  &  85 & $\dots^*$\\[1pt] \hline\noalign{\vglue3pt}
\multicolumn{6}{l}{\parbox{125mm}{\footnotesize\baselineskip=9pt$^*$\,For ${\cal N}\geq5$ there are too many Adinkra topologies to be shown here; see {\scriptsize\tt http://www.rlmiller.org/de\_codes/} for an up-to-date table with links to actual codes.}}
\end{tabular}\vspace{-3mm}
\end{center}
\caption{Minimal Off-Shell 4D, ${\cal N}\le 8$ Supermultiplets}
\label{4DminREPsiz}
\end{table}

Table~\ref{4DminREPsiz} can be used to make an argument about the size of the smallest irreducible off-shell representation for a given value of $\cal N$. If the smallest indecomposable representation coincides with the smallest irreducible representation, then for a given value of $\cal N$, the table above determines the smallest off-shell representation. For each value of $\cal N$, it is possible to consider the representation that appears at the lowest level of that column.   All of these topologies are shown together with the number of bosonic and fermionic nodes in Table~\ref{4DminREPsiz}.

For the cases of ${\cal N} = 1\text{ and }2$, the number of bosonic and fermionic degrees of freedom are in agreement with the known minimal off-shell supersymmetrical representations.  The smallest 4D, ${\cal N} = 1$ off-shell representations do indeed consist of 4 bosons and 4 fermions.  In a similar manner, the smallest 4D, ${\cal N} = 2$ off-shell representations do indeed consist of 8 bosons and 8 fermions. The case of  4D, ${\cal N} = 3$ off-shell representations is not so widely known. Nevertheless, W.~Siegel has presented an argument about the off-shell structure of conformal 4D, ${\cal N} = 3$ supergravity that indicates it describes 64 bosons and 64 fermions\cite{rWSON}.  There exist, also, one known off-shell example of a 4D, ${\cal N} = 4$ supermultiplet in Salam-Strathdee superspace. It is the conformal 4D, ${\cal N} = 4$ supergravity supermultiplet field strength\cite{rBRdW} and it consists of precisely 128 bosons and 128 fermions. All of this agrees with the first four `data' points on Table~\ref{4DminREPsiz}.

Precisely when $N$ is a multiple of $8$, the maximum value of $k$ is $N/2$, and these codes are self-dual (where the orthogonal space of the code equals the code; doubly even implies that these codes are self-orthogonal).  Thus, these relate to even unimodular lattices.  Indeed, the case ${\cal N}=2$, or $N=16$, provides the two lattices, well known to string theorists: $E_8\times E_8$, and $SO(32)$, which we call $E_{16}$, since $D_{16}$ fits in the sequence of codes $D_{2n}$ such that $D_{16}\subset E_{16}$:
\begin{equation}
 D_{16}\iff d_{16}:\left[\begin{smallmatrix}
                          1111\,0000\,0000\,0000\\[1pt]
                          0011\,1100\,0000\,0000\\[1pt]
                          0000\,1111\,0000\,0000\\[1pt]
                          0000\,0011\,1100\,0000\\[1pt]
                          0000\,0000\,1111\,0000\\[1pt]
                          0000\,0000\,0011\,1100\\[1pt]
                          0000\,0000\,0000\,1111
                         \end{smallmatrix}\right]
 \qquad\text{\it vs.}\qquad
 E_{16}\iff e_{16}:\left[\begin{smallmatrix}
                          1111\,0000\,0000\,0000\\[1pt]
                          0011\,1100\,0000\,0000\\[1pt]
                          0000\,1111\,0000\,0000\\[1pt]
                          0000\,0011\,1100\,0000\\[1pt]
                          0000\,0000\,1111\,0000\\[1pt]
                          0000\,0000\,0011\,1100\\[1pt]
                          0000\,0000\,0000\,1111\\[1pt]
                          0101\,0101\,0101\,0101
                         \end{smallmatrix}\right].
\end{equation}

It is of interest to consider the final case above, with ${\cal N}=8$, or $N=32$ supersymmetries.  For this case we have $2^{N-16}$ = $2^{16}$ = 65,536 total nodes. The Adinkra associated with this topology has 32,768 bosonic nodes and 32,768 fermionic nodes.  By very different arguments\cite{rGLPR}, the final result of $\min(d_B)=\min(d_F)=32,768$ in this table appeared very early in our considerations of iso-spinning particles  and ``Garden Algebras'' as the spectrum generating algebras of spacetime supersymmetry.

There are 85 different maximal codes in $N=32$ (up to permutation equivalence).  The attempt to classify these began with Conway and Pless\cite{rCP} in 1980, though a correction was found to be necessary by Conway, Pless and Sloane in 1990\cite{rCPS}.  Bilous and van Rees\cite{rBilRees} in 2007 replicated these results by performing a more systematic search and provided the list of all 85 codes on the web-site\cite{rBilReesW}.  It is amusing to note that this was achieved not long ago, that $N=32$ is the upper limit of what is known currently about self-dual doubly even codes, and that 32 is the maximal $N$ needed in applications to superstrings and their $M$- and $F$-theory extensions.  Five of these $85$ codes have minimal weight 8, and it would be interesting to see if these five play a special role in ${\cal N} =8$ supersymmetries in four dimensions.

\section{Toward a Classification of Adinkras and Supermultiplets}
 \label{CliffHanger}
An Adinkra is determined by its chromotopology, a hanging of the vertices, and a choice of which edges are dashed.  We have in this paper described the classification of chromotopologies of Adinkras: they are disjoint unions of connected chromotopologies, each of which is described by a doubly even code.  Conversely, every doubly even code can be used to form an Adinkra.  As per Definition~\ref{dAd}, to go from this to specifying an Adinkra would require a choice of orientations for all the edges and a choice of dashing of each edge.

Concerning orienting the edges, we have discussed in this paper two such choices that are always available: Isoscalar Adinkras and one-hooked Adinkras.  More generally, in Ref.~\cite{r6-1}, we have shown that the choices in orienting edges correspond to the ways of ``hanging'' the vertices of the graph at various heights, subject to certain conditions.  The Isoscalar Adinkra results from hanging the Adinkra on all of the bosons, which are placed all at the same level.  The one-hooked Adinkra results from hanging the Adinkra from a single boson.

So if we describe the various ways of dashing the edges (an effort which will be described in another paper), would this then classify one-dimensional $N$-extended off-shell supermultiplets?  Not quite.  First, there is the question of whether every such supermultiplet comes from an Adinkra.   Our investigations in this area indicate that there do indeed exist non-adinkraic supermultiplets, but even for these supermultiplets, Adinkras turn out to be useful in their description. \cite{MatterN2Matter}  Second, we might ask whether two Adinkras may describe the same supermultiplet.  But if we insist that the Adinkras are one-hooked, then the corresponding representations are isomorphic if and only if they have the same code, as we saw in Theorem~\ref{thm:unique}.  So by one-hooking the Adinkras, we can be sure that these are all different supermultiplets.  Even considering these two issues, what is clear is that there is a surprising wealth of supermultiplets in one dimension as $N$ is increased towards $N=32$.

Finally, we wish to draw attention to the ``degeneracy'' in constructing even the minimal supermultiplets  for various $N$, uncovered by the listing of permutation equivalence classes of doubly even codes in Table~\ref{t:G3}.  Entries in this table that are greater than 1 indicate situations where there is more than one permutation class of doubly even codes that are possible for that value of $N$ and $k$.  This is true even if we restrict our attention to minimal supermultiplets, which have maximal $k$, when $N\ge 10$.  Note that as $N$ increases, this degeneracy increases, certainly for $k$ non-maximal, but also for $k$ maximal.   This clearly illustrates that this ``degeneracy'' amongst even the minimal supermultiplets grows extremely fast with $N$.

\clearpage
\appendix
\section{Deferred Details on Supersymmetry and Adinkras}
 \label{app:SA}
\subsection{Real Coefficients}
\label{app:real}
Real supermultiplets, $\sM=(F_1\6(\t),\cdots,F_m\6(\t))$ consist of real component fields: $\big(F_A(\t)\big)^\dag=F_A(\t)$, and supersymmetry is assumed to preserve this condition.  In this section, we will show that this condition implies that the coefficients $c$ in \Eq{eQB} and \Eq{eQF} are real.

We use the convention whereby $(XY)^\dag=Y^\dag X^\dag$, regardless whether $X$ and $Y$ are bosonic (commuting) or fermionic (anticommuting) objects, as is standard in the physics literature. 

If $\f_A(\t)$ is real, then its supersymmetry transform
\begin{equation}
\d_Q(\e) := -i\e^I\,Q_I
\end{equation}
 must also be real. Applying this to \Eq{eQB} results in
\begin{equation}
 \d_Q(\e)\, \f_A(\t) = -i\e^I\,c\,\ddt^{\,[\f_A]+\frac12-[\j_B]} \j_B(\t).
 \label{edQe}
\end{equation}
 Thus
\begin{equation}
 \Big(\d_Q(\e)\,\f_A(\t)\Big)^\dag
 ~=~i\,\ddt^{\,[\f_A]+\frac12-[\j_B]} \j^\dag_B(\t)\,c^*\,\e^I
 ~=~-i\e^Ic^*\,\ddt^{\,[\f_A]+\frac12-[\j_B]} \j_B(\t).\label{edQe+}
\end{equation}
 Comparing the right-hand sides of Eqs.\eq{edQe} and\eq{edQe+}, we find that
\begin{equation}
c^* =\,c.
\label{eRealC}
\end{equation}%
Thus the coefficients $c$ are real.

\subsection{The Proof that Scaling Factors $c=\pm 1$ are the Only Ones Necessary}
\label{app:pm}
We now will show that the coefficients in an adinkraic supermultiplet can be chosen to be $c=\pm 1$ via a rescaling of fields by real numbers.

\begin{proposition}\label{rescale}
Suppose we have an adinkraic supermultiplet, with $F_1\6(\t),\cdots,F_{2m}\6(\t)$ for component fields.  There is a real rescaling of these component fields so that each non-zero coefficient $c$ in Eqs.\eq{eQB} and\eq{eQF} is equal to $1$ or $-1$.
\end{proposition}

\begin{proof}
We suppose we have component fields $F_1\6(\t),\cdots,F_{2m}\6(\t)$,\footnote{See footnote~\ref{introf} in the proof of Theorem~\ref{T:QuoC} about the use of this notation.} and supersymmetry generators $Q_I$ so that the supersymmetry transformation rules are all of the form
\begin{equation}
Q_I F_A = c\,\ddt^\lambda F_B\label{adinkralike}
\end{equation}
where $\lambda$ is either $0$ or $1$, and $F_B$ is another component field, and $c$ is a complex number (real if $F_A$ is bosonic, pure imaginary if $F_A$ is fermionic).

We can draw an Adinkra-like graph for this supermultiplet: we again create vertices $v_1,\cdots,v_{2m}$ to correspond to the component fields $F_1,\cdots,F_{2m}$---white for bosonic fields and black for fermionic fields.  For each supersymmetry transformation like the one above, we draw an edge from the vertex $v_A$ to $v_B$.  The arrow goes from $v_A$ to $v_B$ if $\lambda=0$, and the reverse if $\lambda=1$.  As before, this is consistent with the corresponding equation for $Q_I F_B$.  Instead of choosing between a dashed or solid line, we label the edge by $c$ if $F_A$ is bosonic and $i/c$ if it is fermionic.  It is easy to check that this is consistent with the corresponding equation for $Q_I F_B$.

We can work on each connected component of this Adinkra-like graph separately.  So for the following we assume our Adinkra-like graph is connected.

We wish to rescale the various component fields so that $c$ has absolute value 1.  A naive approach would be to simply rescale $F_B$ by $c$, so that $c$ disappears in \Eq{adinkralike}.  But remember that there are $2mN$ of these equations, and only $2m$ component fields.  More visually, imagine the Adinkra-like graph with separate $c$-labels on each edge.  We can imagine starting at one vertex, then going along an edge to the next vertex, rescaling the corresponding field by whatever it takes to make $c=1$ on that edge.  The trouble is that when we come back to a vertex we already saw, we are no longer free to rescale that field without messing up other $c$'s.  But as we will see, it turns out that the supersymmetry algebra guarantees that when we return to a previously-seen vertex, the edge will have $|c|=1$.

We first begin by removing all the loops, using the following standard procedure from graph theory: if we choose an edge, and erase it from the graph, the graph may either separate or stay connected.  We first find an edge so that erasing it keeps the graph connected.  By successively finding such edges and erasing them, we obtain a minimal graph so that this graph is still connected.  This resulting graph is a tree, so that for every pair of vertices $v_i$ and $v_j$, there exists a unique, non-backtracking path from $v_i$ to $v_j$ in this tree (were it not unique, we could remove another edge).  This tree is called a {\em\/spanning tree} for the original graph.  The spanning tree for a graph is not unique, but we select one.  See Figure~\ref{fig:spantree}.

\begin{figure}[ht]
\begin{center}
\begin{picture}(130,75)(0,0)
 \put(1,2){\includegraphics{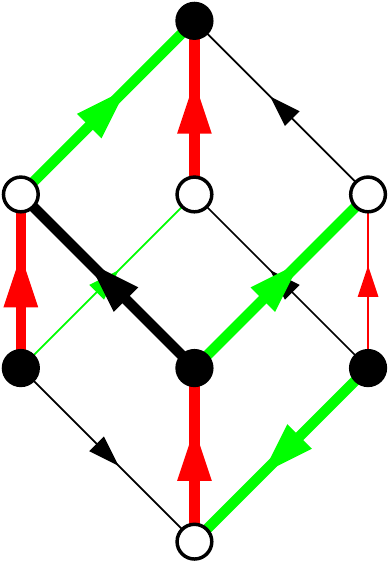}}
 \put(80,2){\includegraphics{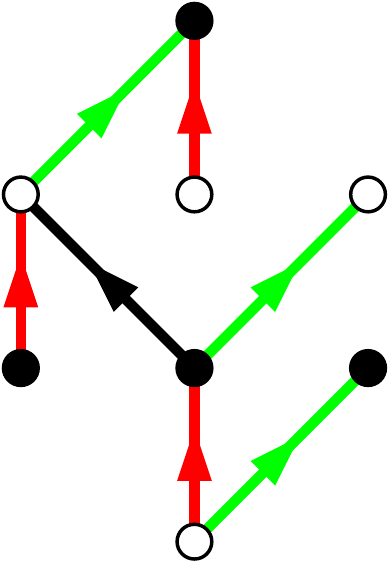}}
\end{picture}
\end{center}
\caption{On the left is an $N=3$ cubical Adinkra.  The dashedness of the lines are suppressed for clarity.  Some edges (shown with thin lines) can be deleted, until what is left is a tree, called the {\em\/spanning tree} (shown at right).\label{fig:spantree}}
\end{figure}

Now pick any vertex in our graph, denote it $v_*$ and the corresponding component field $F_*(\t)$.  For simplicity of exposition we can assume $v_*$ is bosonic.  If $v_i$ is any other vertex, there is a sequence of edges in the tree connecting $v_*$ with $v_i$, corresponding to a sequence of $Q_I$'s, say $Q_{I_1},\cdots, Q_{I_n}$.  We then apply \Eq{adinkralike} iteratively to $F_*(\t)$, obtaining
\begin{equation}
 Q_{I_n}\cdots Q_{I_1} F_*(\t)
  = C_i\,\ddt^{\lambda} F_i(\t) \label{treeAd}
\end{equation}
for some non-negative integer $\lambda$ (recording the number of arrows that point in the opposite direction of the path), and some non-zero coefficient $C_i$.  This coefficient $C_i$ is real if $n\equiv 0, 1\pmod{4}$, and pure imaginary if $n\equiv 2, 3\pmod{4}$.  It is then possible to rescale $F_i(\t)$ as
\begin{equation}
\tilde{F}_i(\t):= |C_i|\,F_i(\t),\label{eqn:rescale}
\end{equation}
and using $\tilde{F}_i(\t)$ instead of $F_i(\t)$ in the description of the supermultiplet. This rescaling is by a real number. Do so for every $v_i\neq v_*$. We have now done all the rescaling of the fields that we will do. The result is a new Adinkra-like graph, where each $C_i$ has absolute value 1.

If instead we want to go from $v_i$ back to $v_*$ in the tree, we can do the supersymmetry transformations in reverse:
\begin{equation}
Q_{I_1}\cdots Q_{I_n} F_i = \frac{1}{C_i}\,\ddt^{n-\lambda} F_*. \label{treeinvAd}
\end{equation}

Now suppose we choose two vertices $v_i$ and $v_j$ of the graph.  We construct a path $P$ from $v_j$ to $v_i$ by first taking the non-retracing path in the spanning tree from $v_j$ to $v_*$, followed by the non-retracing path in the spanning tree from $v_*$ to $v_i$.  See Figure~\ref{fig:treejtoi}.  Note that $P$ might be partially retracing, if the last few edges of the first path is retraced backward by a beginning segment of the second path.  This retracing is immaterial for us.

\begin{figure}[ht]
\begin{center}
\begin{picture}(38,50)(-6,5)
 \put(-3.5,-1.5){\includegraphics{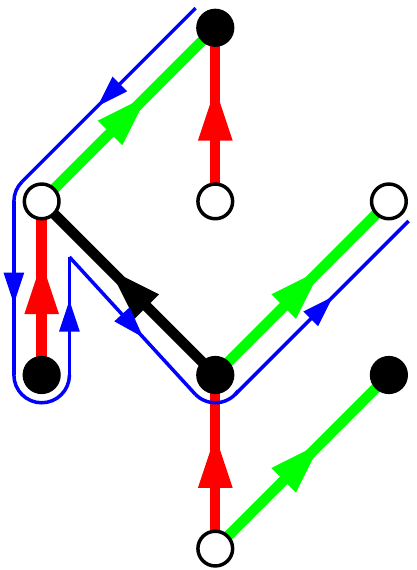}}
 \put(0,12.5){\makebox[0in][r]{$v_*$}}
 \put(39,35){\makebox[0in][l]{$v_i$}}
 \put(23,54){\makebox[0in][c]{$v_j$}}
 \put(2,48){{\color{blue} $P$}}
\end{picture}
\end{center}
\caption{Using the same example as in Figure~\ref{fig:spantree}, suppose we are considering the vertices $v_i$ and $v_j$ as seen in the figure.  The blue path shows a path $P$ in the spanning tree going from $v_j$ to $v_*$, followed by a path in the spanning tree going from $v_*$ to $v_i$.  One edge is retraced, but this is irrelevant for our purposes.  The sequence of colors (black\,=\,1, red\,=\,2, green\,=\,3) along $P$ indicate that we should consider $Q_3Q_1Q_2Q_2Q_3 v_j$ which will be a constant multiple of $\ddt^3 v_i$.  Because of our rescaling in \Eq{eqn:rescale}, we know this constant has absolute value $1$.
\label{fig:treejtoi}}
\end{figure}
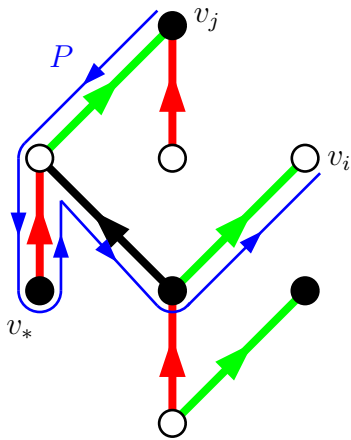

If the path $P$ involves the sequence $Q_{J_1}\cdots Q_{J_k}$ of supersymmetry generators, we then use \Eq{treeAd} and \Eq{treeinvAd} to get
\begin{equation}
Q_{J_k}\cdots Q_{J_1} F_j=\frac{C_i}{C_j}\,\ddt^\mu F_i,\label{eqn:treebothAd}
\end{equation}
where $\mu$ is some non-negative integer, measuring the number of arrows pointing against the path $P$.  Note that $C_i/C_j$ has absolute value $1$.

In particular, if we consider any edge of our Adinkra that is in the spanning tree, its label $c$ now has absolute value $1$.

We now examine an edge of the Adinkra that is not in the spanning tree.  Suppose it connects a bosonic vertex $v_i$ to a fermionic vertex $v_j$, \ie, there is a $Q_I$ so that
\begin{equation}
Q_I v_i=c\,\ddt^{\lambda} v_j\label{eqn:newedge}
\end{equation}
for some real $c$ and $\lambda=0$ or $1$.  Now take the path $P$ in the spanning tree described above, namely, the one that goes from $v_j$ to $v_*$ then to $v_i$.  This corresponds to a sequence of supersymmetry generators $Q_{I_1},\cdots,Q_{I_n}$.  We now close the path $P$ into a closed loop, $P'$ by adding to $P$ the new edge going back from $v_i$ to $v_j$.

\begin{figure}[ht]
\begin{center}
\begin{picture}(38,60)(-6,-4)
 \put(-3.5,-1.5){\includegraphics{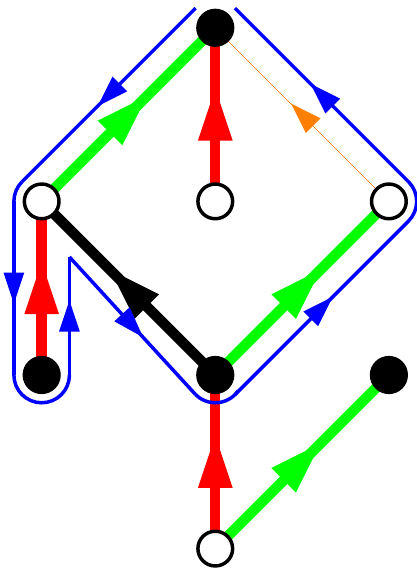}}
 \put(0,12.5){\makebox[0in][r]{$v_*$}}
 \put(39,35){\makebox[0in][l]{$v_i$}}
 \put(18,58){\makebox[0in][c]{$v_j$}}
 \put(33,48){{\color{blue} $P'$}}
\end{picture}
\end{center}
\caption{Suppose we consider the orange edge from $v_i$ to $v_j$, which is in the original Adinkra but not in the spanning tree.  We draw that path $P$ from $v_j$ to $v_i$ in the spanning tree as in Figure~\ref{fig:treejtoi}, and then tack on the orange edge.  The resulting path $P'$, shown in blue, is a loop from $v_j$ to itself.
\label{fig:treeloop}}
\end{figure}
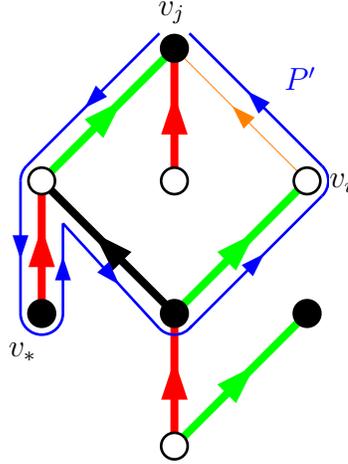

Using \Eq{eqn:newedge} together with \Eq{eqn:treebothAd} we get, for some $r\in\IN$,
\begin{equation}
 Q_I Q_{I_n}\cdots Q_{I_1} F_j(\t)
 = c(-i)^r\frac{C_i}{C_j}\,\ddt^{\lambda+\mu} F_j.
\end{equation}
As this is a closed loop, we can apply this equation to itself, and get
\begin{equation}
 \left(\,Q_I Q_{I_n}\cdots Q_{I_1}\,\right)^2 F_j(\t)
   = (-1)^rc^2\left(\,\frac{C_i}{C_j}\,\right)^2
       \ddt^{2(\lambda+\mu)} F_j(\t).\label{eSqSeq}
\end{equation}
Using \Eq{eSuSy} we can anti-commute the $Q_I$'s past each other, contracting the repeated $Q_I$'s, until we get
\begin{equation}
 \pm (i\ddt)^{\,n+1} F_j(\t)
   = (-1)^r c^2\left(\,\frac{C_i}{C_j}\,\right)^2
       \ddt^{2(\lambda+\mu)} F_j(\t).\label{eResult}
\end{equation}
The $\pm$ sign is determined by $n$ and how many repetitions there are in the $Q_I$ sequence.  Note that $n+1$ is even, since each edge connects vertices of opposite statistics, and a sequence of $n+1$ of these edges go from $v_j$ to itself.  Thus, the left hand side of \Eq{eResult} is actually real.  We can then match coefficients in \Eq{eResult} to get
\begin{equation}
c^2 = \pm (-1)^{r+(n+1)/2}\frac{C_j^2}{C_i^2} = \pm 1.\label{eSign}
\end{equation}
From this we see that $c$ must have complex absolute value $1$.  Since $c$ is real, $c=\pm 1$.
Furthermore, we also note from \Eq{eResult} that $2(\lambda+\mu)=n+1$, indicating that the loop $P'$ involved precisely $(n+1)/2$ arrows going against the path, and thus, $(n+1)/2$ arrows going along the path. This implies that the absence of central charge excludes Escheric loops\cite{rA}.
\end{proof}

This also proves that Adinkraic supermultiplets with no central charge are {\em\/engineerable\/}\cite{r6-1}: it is possible to assign engineering dimensions to all component fields, consistently with the formulas for the supersymmetry generator action of the form Eqs.\eq{eQB1}--\eq{eQF2} of the supersymmetry algebra\eq{eSuSy}, and without having to introduce parameters of nonzero engineering dimension in either of these.

\section{Computing Doubly Even Codes}
 \label{app:RM}
 For some classification problems, there are theorems that answer the question once and for all.  For example, the theorem on the classification of finitely generated abelian groups describes the isomorphism class of such a group via a simple sequence of integers.  In most cases in combinatorics, one cannot hope for such a theorem.  Instead one must settle for an exhaustive enumeration up to a certain size.  It turns out that there are more than 1.1 trillion doubly even $[32,k]$ codes, which without compression will take some number of terabytes to store.  Using methods including those described in this appendix, we have already computed over 60,000 of these codes.
 
 The first section of this appendix gives the formula describing the number of codes we are interested in and a lower bound for the number of permutation-isomorphism classes of codes.  The next section explains the computations necessary for computing the automorphism group and canonical representative of a code, which are critical for the exhaustive enumeration.  In the third section, some optimizations for the particular situation at hand are discussed, which are used in an implementation written specifically for this purpose. The fourth section describes the overarching algorithm designed to produce the intended enumeration.
 
 The methods used to tackle this problem, namely partition refinement and canonical augmentation, were originally developed by McKay\cite{pgi,orderly}, who describes an effective method for computing the automorphism group of a graph and for generating a unique representative for each isomorphism class.  These methods were adapted by Leon in Ref.\cite{leon-main} to tackle similar problems, including several  group theory questions (which are graph isomorphism complete) as well as the problem at hand, namely computing the automorphism group of a linear code.  McKay and Leon have both written optimized software packages implementing these methods.  The first software package, which remains the standard program for determining graph isomorphism, is called \verb|nauty|, and is available online through McKay's website\cite{bdm-url}.    The second is now licensed under the GPL, thanks to the efforts of David Joyner and Vera Pless, and is available as a part of GUAVA\cite{guava}, a GAP package for coding theory. GPL implementations for the cases of graphs of any size and codes of length up to 32 (or 64 on many machines) can be found in Sage\cite{sage} ({\small\verb|sage.graphs.graph_isom|} and {\small\verb|sage.coding.binary_code|}).  For more details about the latter, written specifically for the classification of Adinkra topologies, the Reader can consult Section \ref{new_soft}.  The second and fourth sections are essentially a summary of the techniques described in Refs.\cite{pgi} and\cite{orderly}, respectively, which provide the full story.  The Reader should also be aware of the improvements to this algorithm in the sparse case in Ref.\cite{saucy}.

\subsection{The Mass Formula}
\label{sec:gaborit}
 The following formula, due to Gaborit\cite{rPGMass}, gives the number $\sigma(N,k)$ of 
distinct doubly even $[N,k]$ codes, depending on the residue of $N$ modulo 8:
\begin{align}
  \sigma(N,k) &= 
 \begin{cases}
  \ddd\prod_{i=0}^{k-1}\frac{2^{N-2i-2}+2^{\left\lfloor\frac{N}{2}\right\rfloor-i-1}-1}{2^{i+1} - 1},
      &\quad\text{if $N \equiv 1,7 \pmod 8,$} \\[6mm]
  \ddd\prod_{i=0}^{k-1}\frac{\bigl(2^{\frac{N}{2}-i-1}+1\bigr)\bigl(2^{\frac{N}{2}-i-1}-1\bigr)}{2^{i+1} - 1},
      &\quad\text{if $N \equiv 2,6 \pmod 8,$} \\[6mm]
  \ddd\prod_{i=0}^{k-1}\frac{2^{N-2i-2}-2^{\left\lfloor\frac{N}{2}\right\rfloor-i-1}-1}{2^{i+1} - 1},
      &\quad\text{if $N \equiv 3,5 \pmod 8,$} \\[6mm]
  \ddd\prod_{i=0}^{k-2}\frac{2^{N-2i-2}+2^{\frac{N}{2}-i-1}-2}{2^{i+1} - 1}
     \cdot\vp^+(N,k),
      &\quad\text{if $N \equiv 0 \pmod 8,$} \\[6mm]
 \ddd\prod_{i=0}^{k-2}\frac{2^{N-2i-2}+2^{\frac{N}{2}-i-1}-2}{2^{i+1} - 1}
     \cdot\vp^-(N,k),
      &\quad\text{if $N \equiv 4 \pmod 8,$}
 \end{cases}\\[4mm]
 \vp^\pm(N,k)
     &\Defl\frac{1}{2^{k-1}} + \frac{2^{N-2k} \pm 2^{\frac{N}{2} - k} - 2}{2^k - 1}~,
\end{align}
We are interested in those codes with $N\leq 32$. According to the mass formula above, we find the most such codes when $N=32$ and $k=10$.  In this case,
\begin{equation}
 \begin{aligned}
  \sigma(32, 10) &= 162,953,548,221,364,911,292,708,847,668,107,902,745,573,601,875 \\
   &\approx 1.6 \times 10^{47}.
 \end{aligned}
\end{equation}
However, we are interested not in the codes themselves, but rather their equivalence classes under permutations of the columns.
To establish a lower bound on the number of equivalence classes, suppose that every such code has a trivial automorphism group and thus its orbit is as large as possible, namely $32!$. Hence there are at least 
\begin{equation}
\bigg\lceil\frac{\sigma(32, 10)}{32!}\bigg\rceil = 619,287,158,132
\end{equation}
distinct classes of $[32, 10]$ codes.  If we carry out a similar computation for each case of interest, we find that there are at least
\begin{equation}
\sum_{k=0}^{\vk(32)}\bigg\lceil \frac{\sigma(32,k)}{32!}\bigg\rceil=1,117,005,776,858
~\approx~1.1\times10^{12}
\end{equation}
distinct classes (note that for smaller $N$, we simply add zero columns to obtain an equivalent code of degree 32).  Since the automorphism group of a generic code is small, this lower bound is a reasonable estimate for the actual number of codes.

 Regardless of the speed at which an individual processor can generate codes, the procedure for generating the codes can be run in many parallel processes.  We can not only divide by the rate at which we can generate codes per processor, but we can also divide by the number of processors working on the job.  This is what makes the computation feasible; speed increases almost linearly as a function of the number of processors.

\subsection{Computing the Automorphism Group}
 Computing the automorphism group of a code is closely related to the graph isomorphism problem.  Given a code $C \subset \{0,1\}^N,$ define a bipartite graph $G(C)$, with the vertices partitioned into ``left'' and ``right'', as follows:  The words of the code are the left vertices, and the set $\{1, 2, ..., N\}$ forms the right set of vertices.  Given a word $w$ on the left, and a number $j$ on the right, there is an edge between the two if and only if $w$ has a 1 in the $j^\text{th}$ place.  Consider the automorphism group $\Aut(G(C))$ of the graph, but allow only those permutations that map left vertices to left vertices and right vertices to right vertices.  This gives a subgroup $\Aut(G(C))_b$ which is isomorphic in a canonical way to the permutation automorphism group of the code, simply by considering each permutation's action on the right set of vertices, which is identified with the set of columns of the code.  Our approach is to think of the codes in terms of their corresponding bipartite graphs, with additional structure coming from the linear, self-orthogonal code. The algorithm described in this section computes the automorphism group of the code, as well as a canonical representative of the code, which is an arbitrary but fixed representative of the isomorphism class.
 
 Suppose $G$ is a graph, and $\Pi$ is a partition of the vertices $V(G)$.  The partition $\Pi$ is called an \em equitable partition \em if for every pair of cells $C_1, C_2$ in $\Pi$, the number of edges $\{u, v\}$ such that $v \in C_2$ is constant as $u$ ranges over vertices in $C_1$.  By considering the orbits of vertices of the graph under any subgroup of the automorphism group, one obtains a partition which is always equitable.  However, the converse is not always true; there are equitable partitions which do not arise in this way.  One partition $\Pi_1$ is \em coarser \em than another partition $\Pi_2$ (or $\Pi_2$ is \em finer \em than $\Pi_1$) if every cell of $\Pi_2$ is a subset of a cell of $\Pi_1$.  The \em discrete partition \em is the one where every cell is of size one, and the \em unit partition \em is the one where there is only one cell.
 
 Given a graph $G$ and a partition $\Pi$ of $V(G)$, denote by $E_G(\Pi)$ the coarsest equitable partition of $G$ which is finer than $\Pi$.  In particular, if $\Pi$ is equitable with respect to $G$, then $E_G(\Pi) = \Pi$.  The algorithm described in this section takes as input a graph $G$ and a partition $\Pi_0$ of $V(G)$, and it returns the subgroup $\Aut(G)_{\Pi_0}$ of the automorphism group consisting of permutations that respect $\Pi_0$, \ie, that do not carry any vertex of one cell of $\Pi_0$ into another.  From here on, consider all partitions to be ordered, so that for example $(\{1,2,3\},\{4,5\}) \neq (\{4,5\},\{1,2,3\})$.  In particular, a discrete ordered partition is simply an ordering on the vertices of $G$.  For $v \in V(G)$, if $\Pi = (C_1, ..., C_r)$ and $v \in C_i$, define $R(G, \Pi, v) = E_G(\Pi^\prime)$, where $\Pi^\prime = (C_1, ..., C_{i-1}, \{v\}, C_{i-1}\setminus\{v\}, C_{i+1}, ..., C_r)$.  This defines $R(G, \Pi, v)$ up to reordering of the cells of the partition. For full details, see Ref.\cite{pgi}, Algorithm 2.5, which defines the ordering.
 
 Define a rooted tree $T = T(G, \Pi_0)$ consisting of equitable partitions of $G$ finer than $\Pi_0$ as follows: The root of $T$ is the partition $E_G(\Pi_0)$, and for any node $\Pi$ of $T$, its children are the partitions $R(G, \Pi, v)$ for which $v$ is not yet in a singleton cell of $\Pi$.  The group $\Aut(G)_{\Pi_0}$ acts on the tree $T$ by taking a sequence of nested partitions $\Pi_0, ..., \Pi_k$ to the resulting nested sequence of partitions by letting $\Aut(G)_{\Pi_0}$ act on the elements of the cells; since this sequence is defined by a sequence $v_1, ..., v_k$, the sequence $\gamma(v_1), ..., \gamma(v_k)$ defines the image under $\gamma \in \Aut(G)_{\Pi_0}$.  Further, the subgroup of permutations respecting a partition $\Pi$ acts on any subtree of $T$ rooted at $\Pi$.  Because this action is faithful, the structure of $T$ can be used to calculate a set of generators for $\Aut(G)_{\Pi_0}$.
 
 The algorithm itself is a backtrack algorithm that successively refines the partitions to explore the tree $T$.  Any leaf of the tree is a discrete ordered partition $(\{v_{1}\}, \{v_{2}\}, ..., \{v_{n}\})$, which defines an ordering of the vertices of $G$.  Given two leaves of the tree $T$, one has two orderings $v_1, ..., v_n$ and $v_1^\prime, ..., v_n^\prime$, which we think of as a permutation defined by $v_i \mapsto v_i^\prime$.  This is the means by which the algorithm finds automorphisms, and it uses the presence of automorphisms to deduce when different parts of the tree are equivalent.  At this point the algorithm backtracks towards the root until there is a new part of the tree to explore which is not yet known to be equivalent to a part of the tree already traversed.  In practice, the part of the tree traversed is much smaller than the entire tree, and once one backtracks off of the root, one has a set of generators for $\Aut(G)_{\Pi_0}$.  Invariants and orderings can also be used to find a leaf of the tree $T$ which is maximal amongst the nodes with largest invariants, which is uniquely defined independently of reordering the inputs.  This leaf is a ``canonical label'' for the pair $(G, \Pi_0)$.  In particular, if $\gamma \in \Aut(G)_{\Pi_0}$, then the canonical label for $(G^\gamma, \Pi_0^\gamma)$ is the same as that of $(G, \Pi_0)$.  For details, the Reader is once again directed to the theory in Ref.\cite{pgi} and the code in Ref.\cite{sage}.
 
 The outcome is an algorithm to efficiently compute the group $\Aut(C) = \Aut(G(C))_b$ (here ``='' means canonically isomorphic) as well as a unique representative of each isomorphism class which we will denote $c(C)$, both of which will be used in the algorithm to generate all the permutation classes of doubly even codes.  Also note that by the use of the bipartite graph construction, not only is $c(C)$ a code, but it also contains an ordering of its words.  In practice, $c(C)$ will be output as an actual generator matrix.

\subsection{New Software for Special Circumstances}
\label{new_soft}
 It is a curious coincidence that the limit on $N$ for this application is 32.  This implies that the size of the words in a code are at most the size of the machine words on a 32-bit system.  The analogy goes further: the operation of adding two vectors in a code becomes the single clock tick of taking an \verb|XOR| on the machine words representing them.

 When analyzing codes in terms of their corresponding graphs, we can take advantage of the linear orthogonal structure of the codes to optimize standard graph algorithms.
For example, given an equitable partition $\Pi$ of the bipartite graph $G(C)$, the set of words appearing in their own singleton cells of $\Pi$ is closed under binary addition.  This is immediate from the definition of equitable partition and of binary addition, and this fact greatly reduces the expected height of the tree to be searched in computing $c(C)$.  Further, every permutation acts linearly, so to check if some permutation is indeed an automorphism, we need only check it on a basis.  It is also possible to use the linear structure to derive variants on the refinement procedure, depending on the size and type of codes being generated.
 
 In the generation algorithm, since we are interested only in self-orthogonal codes, we can use more linear algebra to efficiently enumerate the possible children of a code. We start by writing the code in standard form, $C = [I_k|C^\prime]$, and then extend the basis given by the rows to a basis for $C^\perp$.  This basis can be taken so that the first $k$ positions of each additional vector are zero, and we can also shuffle the pivot columns to the front, so that we obtain this basis as the rows of a matrix of the form $$\left[\begin{array}{c|c}I_k & C^{\prime\prime} \\\hline0 & \begin{array}{c|c}I_{N-2k} & *\end{array}\end{array}\right],$$ where $C^{\prime\prime}$ is a rearrangement of the columns of $C^\prime$.  The rowspan of the bottom part of the matrix, $\left[\begin{array}{ccc}0 & I_{N-2k} & *\end{array}\right]$, then forms a set of unique coset representatives of the original code.  Furthermore, (\ref{eqn:inclusionexclusion}) implies the following:
\begin{corollary} Suppose $C$ is a code spanned by $\{v_1,\cdots,v_k\}$, and that $\wt(v_i) \equiv 0 \pmod 4$ for each $i = 1,\dots,k$. If either of the following hold for all $i \neq j$, then $C$ is doubly even:
\begin{itemize}
\item[a\em)] $\langle v_i, v_j \rangle = 0$,
\item[b\em)] $\wt(v_i + v_j) \equiv 0 \pmod 4$.
\end{itemize}
\end{corollary}
\begin{proof} These results are special to binary codes, since every vector in $C$ is the sum of distinct vectors in $\{v_1,\cdots,v_k\}$. Equation (\ref{eqn:inclusionexclusion}) implies that both conditions are equivalent, so we use {\em b}). Suppose $x,y,z \in C$ are such that
\[\wt(x) \equiv \wt(y) \equiv \wt(z) \equiv \wt(x+y) \equiv \wt(x+z) \equiv \wt(y+z) \equiv 0 \pmod 4.\]
Then again by (\ref{eqn:inclusionexclusion}),
\begin{align*}
\wt(x + y + z) \equiv&\ \wt(x + y) + \wt(z) - 2\langle x+y, z \rangle\\
\equiv& - 2\langle x, z \rangle - 2\langle y, z \rangle\\
\equiv& - 2\wt(x) - 2\wt(z) + 2\wt(x+z)\\
& - 2\wt(y) - 2\wt(z) + 2\wt(y+z)\\
\equiv&\ 0 \pmod 4.
\end{align*}
By induction on the number of basis elements in a linear combination, $C$ is doubly even.
\end{proof}
Going back to the situation at hand, this means that we need only examine the doubly even vectors in the rowspan of $\left[\begin{array}{ccc}0 & I_{N-2k} & *\end{array}\right]$, since the code we already have is doubly even and self-orthogonal, and any vector to be considered is already orthogonal to the code.  These vectors are in one-to-one correspondence with the set of doubly even codes of one dimension higher containing (our rearrangement of) $C$ .  We can do even better by starting with the even subcode of $\left[\begin{array}{ccc}0 & I_{N-2k} & *\end{array}\right]$, since the sum of two even words is even.

\subsection{Exhaustively Generating the Codes}
 Here we use another algorithm developed by McKay\cite{orderly}, called canonical augmentation.  The ingredients for this are a question like the one we are considering, a canonical labeling function as described, a way of computing the automorphism group of a code, and a hereditary structure on the objects to be generated.
 
 We obtain a hereditary structure by thinking of a certain code's \em children \em as the set of codes constructible by adding a single word not already in the code. Thus the dimension of a code is one less than the dimension of all its children, and every code can be built up from the zero dimensional code by a limited number of augmentations.  In our application, this puts all the desired codes on a tree of height 16.  At this point one can consider the very naive algorithm of generating all the children for each code, keeping a list and throwing out isomorphs.  However, the obvious problem with this approach is the excessive isomorphism computation, which is expensive.  This problem is not really ameliorated by storing the canonical representative of each code, as the canonical representative calculation is almost as expensive as isomorphism testing.

 The idea behind canonical augmentation is that instead of requiring the objects generated be in canonical form, we require simply that the augmentation itself be canonical.  The definition of an \em augmentation \em is simply an ordered pair $(C, C^\prime)$, where $C^\prime$ is a child of $C$.  An isomorphism of augmentations is a permutation $\gamma$ such that $(C^\gamma,{C^\prime}^\gamma) = (D,D^\prime)$.  For example, $(C, C^\prime) \cong (C, C^{\prime\prime})$ implies that there is a $\gamma \in \Aut(C)$ such that ${C^\prime}^\gamma = C^{\prime\prime}$.
 
 Suppose we have a function $p$ which takes codes $C$ of dimension $k > 0$ to codes $p(C)$ of dimension $k-1$, such that $C$ is a child of $p(C)$.  Suppose further that $p$ satisfies the following property: if $C \cong D$, then $(p(C), C) \cong (p(D), D)$.  If we have such a function, we can define that an augmentation $(C, D)$ is canonical if $(C, D) \cong (p(D), D)$.  Now suppose that we have two augmentations $(C, C^\prime)$ and $(D, D^\prime)$ such that $C^\prime \cong D^\prime$.  If they were both canonical augmentations then we would have
 $$(C, C^\prime) \cong \bigl(p(C^\prime), C^\prime\bigr) \cong \bigl(p(D^\prime), D^\prime\bigr) \cong (D, D^\prime).$$
  In other words, if both augmentations are canonical, then $C^\prime \cong D^\prime$ implies that $C \cong D$.  Thus any repeated isomorphs would have to be due to the parents being repeated.  If we assume that there are no isomorphs on the parent level, then this implies that $C = D$, and $(C, C^\prime) \cong (C, D^\prime)$.  In other words, under the framework of canonical augmentation, repeated isomorphs arise only due to automorphisms of the parents.  Thus we have Algorithm \ref{alg_gen} below.
 \begin{algorithm}
 \caption{Generate all doubly even codes of degree 32, isomorph-free.}
 \label{alg_gen}
 \begin{verbatim}

 C_0 := the dimension zero code, of degree 32
 traverse(C_0)
 procedure traverse(C):
     report C
     children := the children of C
     representatives := {}
     for D in children:
         if D is minimal in its orbit under Aut(C):
             add D to representatives
     for D in representatives:
         if (C, D) is a canonical augmentation:
             traverse(D)
\end{verbatim}\end{algorithm}

The only remaining question is to produce such a function $p$.  This is where we use the canonical label defined in the last section.  If $C$ is a code of dimension $k$, let $\gamma$ be the permutation taking $C$ to $c(C)$.  Recall that $c(C)$ came with a particular generator matrix determined by the ordering of the words of the code.  Removing the last row of that generator matrix gives a code $C^\prime$ of the same dimension as $C$ and applying $\gamma^{-1}$ to that, we arrive at $p(C)$.  Suppose that $C \cong D$, which in particular implies that $c(C) = c(D)$.  Let $\gamma_C, \gamma_D$ be the permutations taking $C,D$ to $c(C), c(D)$, respectively.  These are isomorphisms, so in fact, $\gamma_D^{-1}\circ\gamma_C$ is an isomorphism from $C$ to $D$ taking $p(C)$ to $p(D)$ by the definition of $p$.  This proves the property that $(p(C),C) \cong (p(D),D)$.  Further, since we have already done the computation $c(C)$, we can obtain not only the canonical label but also generators for the automorphism group $\Aut(C)$ for free.  Thus when we are checking whether $(C^\prime, C) \cong (p(C), C)$, we can simply look for an element of $\Aut(C)$ that takes $C^\prime$ to $p(C)$.

 The final task of enumerating the codes for storage to disk will approximate the experience of a harvest.  Many separate worker processes will be working in parallel, each examining the pairs $(C, D)$ as in the second to last line of Algorithm \ref{alg_gen}, for a fixed $C$ and many $D$.  Each worker will receive one code at a time, and perform all the steps in Algorithm \ref{alg_gen}, only instead of recursively calling the function on the last line, they will record the augmented codes that succeed in a table, which will each in turn become the fodder for another worker.  As codes are given to workers to augment on, they will be flagged as searched, and once all the codes have been flagged and all the workers are done, all the desired codes will be in the table.  The intended end result of all this is a searchable online database which categorizes the codes by $N$ and $k$,  automorphism group size, weight distribution, and perhaps other parameters.

\bigskip\bigskip\paragraph{\bf Acknowledgments:}
 The research of S.J.G.\ is supported by the endowment of the John S.~Toll Professorship, the University of Maryland Center for String \& Particle Theory, and National Science Foundation Grant PHY-0354401.  T.H. is indebted to the generous support of the Department of Energy through the grant DE-FG02-94ER-40854, as well as the Department of Physics, University of Central Florida, Orlando FL, and the Physics Department of the Faculty of Natural Sciences of the University of Novi Sad, Serbia, for recurring hospitality and resources.  The work described in Appendix~\ref{app:RM} uses computer facilities supported by the National Science Foundation Grant No.~DMS-0555776.  Some Adinkras were drawn with the aid of the {\em Adinkramat\/}~\copyright\,2008 by G.~Landweber.  A special thanks goes to Martin Mueller, Erik Duamine, Morteza Zadimoghaddam, Richard Stanley, all at M.I.T., and especially Yan Zhang, Ricard Stanley's student, for fruitful discussions.

\vfill

\vfill\clearpage
\bibliographystyle{elsart-numX}
\bibliography{Refs}

\end{document}